\documentclass[10pt]{article}

\usepackage[left=1in, right=1in, top=1in, bottom=1in]{geometry}
\usepackage[utf8]{inputenc}
\usepackage{subfig}

\usepackage{amsthm}
\usepackage{amssymb}
\usepackage{amsmath}
\usepackage{mathtools}
\usepackage{stmaryrd}
\usepackage[sanserif,full]{complexity}

\usepackage{csquotes}
\usepackage{enumitem}   
\usepackage{tikz}
\usepackage{graphicx}
\usepackage{xspace}
\usetikzlibrary{arrows.meta,matrix}
\usepackage[colorlinks=true,allcolors=blue,allbordercolors=white]{hyperref}
\usepackage{caption}  
\usepackage{dirtytalk}
\usepackage{thm-restate}

\setlist[enumerate]{itemsep=\smallskipamount,parsep=0pt,label={\rm \roman*)}}
\setlist[itemize]{itemsep=\smallskipamount,parsep=0pt}

\newcommand{\toc}[1]{\pagenumbering{roman}\setcounter{tocdepth}{#1}\tableofcontents\newpage\pagenumbering{arabic}}

\usepackage[skip=8pt plus2pt]{parskip}
\makeatletter
\def\thm@space@setup{%
  \thm@preskip=\parskip \thm@postskip=0pt
}
\makeatother
\makeatletter
\renewenvironment{abstract}{%
  \if@twocolumn
    \section*{\abstractname}%
  \else
    \small
    \begin{center}%
      {\bfseries \abstractname\vspace{-.5em}\vspace{\z@}}%
    \end{center}%
    \quotation
    \setlength{\parskip}{8pt plus2pt}
  \fi}
  {\if@twocolumn\else\endquotation\fi}
\makeatother

\newcommand{\Ind}[1]{\mathbb{I}\left[#1\right]}
\newcommand{\zo}{\{0,1\}}
\newcommand{\ra}{\rightarrow}

\newcommand{\eps}{\varepsilon}

\newcommand{\cM}{\mathcal{M}}
\newcommand{\cG}{\mathcal{G}}
\newcommand{\cA}{\mathcal{A}}

\newcommand{\vo}{\vec{1}}
\newcommand{\vz}{\vec{0}}

\newcommand{\la}{\leftarrow}

\newcommand{\defn}[1]{{\textit{\textbf{\boldmath #1}}}\xspace}
\renewcommand{\paragraph}[1]{\vspace{0.09in}\noindent{\bf \boldmath #1.}}

\usepackage{bbm}

\DeclareMathOperator{\Span}{Span}

\DeclareMathOperator{\Comp}{Comp}
\DeclareMathOperator{\Decomp}{Decomp}
\DeclareMathOperator{\Fix}{Fix}

\DeclareMathOperator{\Det}{det}
\DeclareMathOperator{\GF}{GF}

\newcommand{\interior}[1]{ {\kern0pt#1}^{\mathrm{o}} }

\newcommand{\N}{\mathbb{N}}

\newcommand{\ROUTE}{\mathsf{ROUTE}}

\usepackage[linesnumbered,noline,ruled]{algorithm2e} \DontPrintSemicolon
\SetKwRepeat{Do}{do}{while}

\newcommand{\paren}[1]{\left( #1 \right)}

\newcommand{\bin}{\{0,1\}}

\usepackage[capitalise,nameinlink,noabbrev]{cleveref}
\crefname{equation}{}{} 
\crefname{enumi}{Step}{} 

\theoremstyle{definition}
\newtheorem{theorem}{Theorem}[section]
\newtheorem{fact}[theorem]{Fact}
\newtheorem{definition}[theorem]{Definition}

\newtheorem{proposition}[theorem]{Proposition}
\newtheorem{lemma}[theorem]{Lemma}
\newtheorem{claim}[theorem]{Claim}
\newtheorem{corollary}[theorem]{Corollary}

\usepackage{mdframed}
\usepackage{framed}

\usepackage{cryptocode}

\def\draft{1}
\newcommand{\nnote}[1]{\ifnum\draft=1\textcolor{purple}{[\textbf{Nathan:} #1]}\fi}
\newcommand{\nathan}[1]{\ifnum\draft=1\textcolor{purple}{[\textbf{Nathan:} #1]}\fi}
\newcommand{\tnote}[1]{\ifnum\draft=1\textcolor{teal}{[\textbf{Ted:} #1]}\fi}
\newcommand{\ted}[1]{\ifnum\draft=1\textcolor{teal}{[\textbf{Ted:} #1]}\fi}
\newcommand{\ian}[1]{\ifnum\draft=1\textcolor{red}{[\textbf{Ian:} #1]}\fi}
\newcommand{\james}[1]{\ifnum\draft=1\textcolor{green!70!black}{[\textbf{James:} #1]}\fi}
\newcommand{\surendra}[1]{\ifnum\draft=1\textcolor{orange}{[\textbf{surendra:} #1]}\fi}
\newcommand{\unfinished}[1]{\ifnum\draft=1\textcolor{orange}{[#1]}\fi}

\newclass{\AVOID}{AVOID}
\newclass{\svAVOID}{svAVOID}
\newclass{\CL}{CL}
\newclass{\searchCL}{searchCL}
\newclass{\CBPL}{CBPL}
\newclass{\FSTP}{FS_2P}
\newclass{\avgZPP}{avgZPP}
\newclass{\FZPP}{FZPP}
\newclass{\STP}{S_2P}
\newclass{\avgP}{avgP}
\newclass{\LOG}{L}
\newclass{\FL}{FL}
\newclass{\FCL}{FCL}
\newclass{\FSPACE}{FSPACE}
\newclass{\FCSPACE}{FCSPACE}
\newclass{\inplaceFL}{inplaceFL}
\newclass{\inplaceFCL}{inplaceFCL}
\newclass{\inplaceFSPACE}{inplaceFSPACE}
\newclass{\inplaceFCSPACE}{inplaceFCSPACE}
\newclass{\LOSSY}{LOSSY}
\newclass{\ClassNP}{NP}
\newclass{\unifNC}{unifNC}
\newclass{\CoNP}{coNP}

\DeclareMathOperator{\Poly}{poly}
\DeclareMathOperator{\negl}{negl}

\newcommand{\ts}{\texorpdfstring}

\newlang{\lossycode}{LossyCode}

\newcommand{\emphdef}[1]{\defn{#1}} 

\usepackage[
    backend=biber,
    style=alphabetic,
    maxnames=10,
    maxalphanames=10, 
]{biblatex}
\title{The Structure of In-Place Space-Bounded Computation}

\newif\ifnames

\namestrue

\ifnames
\author{James Cook\\\texttt{falsifian@falsifian.org} \and Surendra Ghentiyala\thanks{This work is supported in part by the NSF under Grants Nos.~CCF-2122230 and CCF-2312296, a Packard Foundation Fellowship, and a generous gift from Google.}\\Cornell University\\\texttt{sg974@cornell.edu} \and Ian Mertz\\Charles University\\\texttt{iwmertz@iuuk.mff.cuni.cz} \and Edward Pyne\thanks{Supported by an NSF Graduate Research Fellowship.}\\MIT\\\texttt{epyne@mit.edu} \and Nathan Sheffield\\MIT\\\texttt{shefna@mit.edu}}
\fi

\date{}

\begin{document}
\begin{titlepage}
\maketitle
\begin{abstract}
    In the standard model of computing multi-output functions in logspace ($\FL$), we are given a read-only tape holding $x$ and a logarithmic length worktape, and must print $f(x)$ to a dedicated write-only tape. However, there has been extensive work (both in theory and in practice) on \textit{in-place} algorithms for natural problems, where one must transform $x$ into $f(x)$ in-place on a single read-write tape with only $O(\log n)$ additional workspace. We say $f\in \inplaceFL$ if $f$ can be computed in this model.
    
    We initiate the study of in-place computation from a structural complexity perspective, proving upper and lower bounds on the power of $\inplaceFL$. We show the following:
    \begin{itemize}
        \item Unconditionally, $\FL\not\subseteq \inplaceFL$. 
        
        \item For a permutation $f$, if $f\in \inplaceFL$ then $f^{-1} \in\avgP$. Thus, the problems of integer multiplication and evaluating $\NC^0_4$ circuits lie outside $\inplaceFL$ under cryptographic assumptions. 
        
        \item Despite this, evaluating $\NC^0_2$ circuits can be done in $\inplaceFL$.
        \item We have
        $
        \FL \subseteq \inplaceFL^{\STP}.
        $
        Consequently, proving $\inplaceFL\not \subseteq \FL$ would imply $\SAT\notin \LOG$.
    \end{itemize}
    We likewise show several extensions and strengthenings of the above results to in-place \textit{catalytic} computation ($\inplaceFCL$), where the in-place algorithm has a large additional worktape tape that it must reset at the end of the computation:
    \begin{itemize}
        \item Assuming $\CL \neq \PSPACE$, then $\FCL \not\subseteq \inplaceFCL$, and under cryptographic assumptions, integer multiplication and $\NC_4^0$ evaluation lie outside $\inplaceFCL$.
        \item Despite this, $\inplaceFCL$ can provably compute matrix multiplication and inversion over polynomial-sized finite fields.
    \end{itemize}
    
    We use our results and techniques to show two novel barriers to proving $\CL \subseteq \P$. First, we show that any proof of $\CL\subseteq \P$ must be non-relativizing, by giving an oracle $O$ relative to which $\CL^O=\EXP^O$. This answers an open problem raised in the survey of (Mertz B. EATCS).
    Second, we show that a search problem not known to be in $\P$, namely
    $\mathcal{C}$-$\lossycode$ for circuits of small width and depth, is in $\searchCL$.
\end{abstract}

\thispagestyle{empty}
\end{titlepage}

\toc{2}

\section{Introduction}

The standard model for sublinear-space computation when dealing with functions outputting multiple bits is a \defn{transducer}:
the input is provided on a read-only tape, and the output must be written to a write-only tape, with the aid of a (bounded-length) read/write work-tape.
In many theoretical settings, this is the right notion --- for instance, if a problem $A$ has a logspace algorithm and there is a reduction from $B$ to $A$ computable by a logspace transducer, then $B$ has a logspace algorithm.

However, this is not the only definition one could use.
If we restrict our attention to functions which preserve the length of the input, another natural definition of space complexity could be ``given the input on a \emph{read/write} tape, how much additional read/write workspace is needed to mutate the input into the output?''
That is, instead of allowing a separate write-only output tape, we require the algorithm to \emph{replace} the input with the output, using minimal additional space overhead.
This gives rise to an alternative class of functions --- for the purposes of this paper, we will refer to the class of $n$-bit to $n$-bit functions computable by a logspace transducer as \defn{$\FL$}, and the class computable in this in-place fashion with logarithmic space overhead as \defn{$\inplaceFL$}.

As a complexity class, $\inplaceFL$ is in some ways less clean than $\FL$. 
One particularly ``brittle'' feature is that $\inplaceFL$ is highly sensitive to the input encoding: if the encoding is inefficient, it may be possible for an in-place machine to free up a large amount of additional workspace by compressing the input, and hence gain more computational power.
Nevertheless, there are many cases where in-place space complexity is a coherent and important notion. In particular, it is arguably better suited for the study of \emph{applied} small space computation.
There has been substantial work on in-place algorithms for particular problems like list sorting and other array permutations, fast Fourier transforms, and computational geometry, motivated by the real-life memory overhead of processing large datasets (for a discussion of some of this work, see \cref{sec:background}).
Despite this extensive interest in the algorithms community, we believe this is the first work to define and prove general structural results about $\inplaceFL$.

Computing functions in-place has also emerged as a component of complexity theory, in arguments regarding \defn{catalytic computing}.
In the model of catalytic logspace (corresponding to complexity class $\CL$), a machine with a logarithmic-length work tape has additional access to a poly-length ``catalytic'' tape, which starts arbitrarily initialized, and must be reset after the computation.
Several structural results for catalytic computing (for instance, the recent proofs of $\mathsf{BPCL} = \CL$ and $\mathsf{NCL} = \CL$~\cite{cook2025structure,koucky2025collapsing}, or the catalytic algorithm for bipartite matching~\cite{agarwala2025bipartite}) have relied on a ``compress-or-random'' framework: if a segment of the catalytic tape has some unusual property, then the algorithm can compress that portion of the tape, freeing up more work-space to perform the computation.
For these arguments, the relevant computational property of the compression scheme is precisely such an in-place requirement: the algorithm must be able to mutate the segment of catalytic tape to replace it by its compressed form, with only a small amount of additional space overhead.
A major part of these compress-or-random proofs generally consists of demonstrating that the requisite compression can indeed be performed in-place.

With the goal of developing systematic tools to address these problems, we consider the complexity classes \defn{$\FCL$} and \defn{$\inplaceFCL$}, which are defined analogously to $\FL$ and $\inplaceFL$, but with the addition of a $\poly(n)$-bit read/write catalytic tape that must be reset at the end of computation.\footnote{Note that this means $\FL \subseteq \FCL$ and $\inplaceFL \subseteq \inplaceFCL$.}
We study the strength of $\inplaceFCL$, including identifying several linear-algebraic problems solvable in $\inplaceFCL$ but not obviously in $\inplaceFL$.
Finally, we apply our results and techniques here to prove several structural statements about $\CL$.

\subsection{Our Results}

In this work we initiate the systematic and structural study of in-place space-bounded
algorithms, specifically $\inplaceFL$ and $\inplaceFCL$. Our results can broadly
be grouped into three categories:
\begin{enumerate}
    \item separations between the standard and in-place models of space;
    \item exhibiting new and surprising functions computable in-place; and
    \item using in-place computation to show barriers against proving $\CL \subseteq \P$.
\end{enumerate}

\subsubsection{Separations}
Perhaps the most basic structural question for an in-place class would be how it
compares to the standard computation class. For $\FL$ and $\FCL$ we show that
they (unconditionally and conditionally, respectively) are not equal:

\begin{restatable}{proposition}{implaceFLsep}\label{prop:inplacehax}
    $\inplaceFL \not\subseteq \FL$.
\end{restatable}

\begin{proposition}\label{prop:inplacehaxcl}
    Assuming $\CL\neq \PSPACE$, then $\inplaceFCL \not\subseteq \FCL$.
\end{proposition}

We also rule out the reverse inclusions under weak cryptographic assumptions. Let $\unifNC^0_4$ be the set of functions computable by logspace-uniform constant-depth circuits of locality $4$. We show that these functions are unlikely to lie in $\inplaceFL$:

\begin{restatable}{theorem}{owphard}\label{thm:owps-means-inplace-sad}
    If there exists a one-way permutation computable in $\FL$, then
    $\unifNC^0_4 \not\subseteq \inplaceFCL$.
\end{restatable}

Since $\unifNC^0_4 \subseteq \FL \subseteq \FCL$, \cref{thm:owps-means-inplace-sad}
shows that basic cryptographic assumptions imply $\FL \not\subseteq \inplaceFL$ and
$\FCL \not\subseteq \inplaceFCL$. However, we show that an \textit{unconditional} separation would imply breakthrough separations of well-studied classes:

\begin{restatable}{theorem}{TheoremFPInPlace}\label{thm:FPInPlace}
Every length-preserving function \( f \in \FZPP \) can be computed in \( \inplaceFL^{ \NP / \Poly } \cap \inplaceFL^{\STP} \).
\end{restatable}

\begin{corollary}\label{cor:sat-in-l-means-inplace-happy}
    Assuming $\FCL \not\subseteq \inplaceFL$, then $\NP\neq\LOG$.
\end{corollary}

\subsubsection{Algorithms}

Beyond the question of standard versus in-place computation, we additionally
demonstrate several unconditional inclusions of problems and classes in
$\inplaceFL$ and $\inplaceFCL$.
For instance, while evaluating uniform $\NC^0_4 $ circuits is conditionally ruled out for both
classes by \cref{thm:owps-means-inplace-sad}, locality 2 \textit{is} unconditionally computable:
\begin{restatable}{theorem}{nctwoeval}\label{thm:locality-2-inplace}
    $\unifNC^0_2 \subseteq \inplaceFL$.
\end{restatable}
This result directly implies an in-place algorithm for evaluating (arbitrary-depth) uniform circuits of fan-in two and width $n+O(\log n)$ (and see \cref{thm:small-width-computable} for details).

With the addition of catalytic space, we also give in-place algorithms for linear algebra problems, with our most fundamental being matrix-vector product,
and by extension matrix-matrix product:

\begin{restatable}{theorem}{matvecprod}\label{thm:mat-vec-prod}\label{thm:inplace-matrix-vector}
For any field \( K \) representable in space \( O(\log n) \), there is an algorithm in \( \inplaceFL^{ \CL } \subseteq \inplaceFCL\) which, given read-only access to a matrix \( A \in K^{ n \times n } \), replaces a vector \( x \in K^n \) with \( Ax \) in-place. 
\end{restatable}

\begin{restatable}{corollary}{matmatprod}\label{cor:mat-mat-prod}\label{cor:inplace-matrix-matrix}
For any field \( K \) representable in space \( O(\log n) \), there is an algorithm in \( \inplaceFL^{ \CL } \subseteq \inplaceFCL\) which, given read-only access to a matrix \( A \in K^{ n \times n } \), replaces a matrix \( B \in K^{n \times n} \) with \( AB \) in-place. 
\end{restatable}

Furthermore, with some more tricks from catalytic computing, which appear in \cref{app:matrix-compress} and may be of independent interest,
\cref{cor:inplace-matrix-matrix} allows us to invert matrices as well:

\begin{restatable}{corollary}{inplaceinvert}\label{thm:inplace-matrix-invert}
    For any field \( K \) representable in space \( O(\log n) \), there is an algorithm in \( \inplaceFL^{ \CL } \subseteq \inplaceFCL\) which replaces a matrix \( A \in K^{n \times n} \) with \( A^{-1} \) in-place.
\end{restatable}

\subsubsection{Barriers}
Finally, we apply our techniques and results to the standard model of catalytic logspace. We show two barriers to resolving the question of whether
$\CL \subseteq \P$. Recall that $\CL\subseteq \ZPP$ (\cite{buhrman2014computing}), so $\CL \subseteq \P$ holds under standard de-randomization assumptions, but this is not known unconditionally. Prior to this work there were no explicit (non-promise) problems known to be in $\CL$ but not known to be in $\P$. Thus, it was unclear if we should expect $\CL\subseteq \P$ to be a difficult problem to resolve. 
Our first and main result is an oracle such that $\CL^O \not\subseteq \P^O$,
which rules out a relativizing proof of $\CL \subseteq \P$:
\begin{theorem}
    There exists an oracle $O$ such that $\CL^O = \EXP^O$.
\end{theorem}
The existence of an oracle such that $\CL^O = \PSPACE^O$ was demonstrated by Buhrman,
Cleve, Kouck\`y, Loff, and Speelman~\cite{buhrman2014computing},
while Heller~\cite{heller1984relativized} showed an oracle such that $\ZPP^O = \EXP^O$;
our result strengthens and unifies both of these.\footnote{The definition of space-bounded oracles is known to be tricky. Our result (as does the prior result of \cite{buhrman2014computing}) holds with regard to the weakest such model, where we do not allow a separate query tape, and force the catalytic algorithm to only make queries on subsets of its catalytic tapes (and see \cref{subsec: oracle_defs} for details).} This also answers an open question
by Mertz~\cite{mertz2023reusing}.

Second, we show the first explicit example of a problem in $\text{search-}\CL$
which we do not (immediately) know how to compute in $\text{search-}\P$.
Our problem will be the $\lossycode$ problem,
introduced by Korten~\cite{korten2022derandomization} (see \cref{sec:defs}),
as applied to a restricted class of circuits:

\begin{restatable}{theorem}{lossysmallwidth}\label{cor:lossy-for-those-weird-circuits}
    Let $\mathcal{C}$ be the class of fan-in $2$ layered circuits of width $n + O(\log n)$ and depth $O(\log n)$. We have $\mathcal{C}\text{-}\lossycode \in \searchCL$.\footnote{Note that $\searchCL$ means there is relation $R$ on pairs $(x,y)$, and given an input $x$ the catalytic algorithm must write some $y$ for which $(x,y)\in R$ to the output tape. This generalizes $\FCL$, where there is a single valid output for each input $x$.}
\end{restatable}

Assuming $\ZPP = \P$ we immediately get both $\CL \subseteq \P$ and
deterministic algorithms for $\mathcal{C}\text{-}\lossycode$ (in fact, $\P/\poly\text{-}\lossycode$) ---
\cref{cor:lossy-for-those-weird-circuits} implies that the former result
is at least as hard to prove as the latter.

\subsection{Background and related work}\label{sec:background}

\subsubsection{In-place algorithms}

There is a substantial amount of existing literature considering in-place algorithms for particular problems, both in theory and in practice.
The notion of ``in-place'' varies across the literature and does not always coincide with our definition of \inplaceFL.
In particular, some papers use ``in-place'' simply to refer to $\FL$, while some papers allow in-place algorithms a write-only output tape \emph{as well as} write access to the input.
(A notable variant of the latter idea is the ``Restore'' model of Chan, Munro, and Raman~\cite{chan2014selection,dumas2024place2,dumas2024place}, where modifications to the input are allowed only if they are reverted by the end of the algorithm's runtime, a condition reminiscent of catalytic computing.)
Let us briefly mention several existing lines of work that follow our definition of in-place: namely, where the input must be overwritten by the output.

\textbf{Sorting and array permutations.}
In-place computation has been a particular focus for sorting algorithms. 
It is a classic result that an $n$-entry array of $O(\log n)$-bit integers can be sorted in-place with $O(\log n)$ space overhead in $O(n \log n)$ time --- for instance, via heapsort.
There has since been work giving in-place versions of mergesort~\cite{merge-partition,merging2,cache-oblivious-mergesort} and radix sort~\cite{franceschini2007radix}, minimizing the number of moves needed for in-place sorting~\cite{kindafewmoves,fewmoves, superfewmoves}, and optimizing in-place sorting for practical hardware concerns~\cite{peters2010fast,awan2016gpu}.
In addition to sorting, time-space tradeoffs for other in-place list problems have been considered, especially the problem of applying an arbitrary permutation (provided as auxiliary read-only input)~\cite{melville1979time,fich1995permuting,el2016raising,guspiel2019place,dudek2021strictly}.
There has also been recent work studying these problems in the ``Parallel In-Place'' model, whereby the in-place computation is distributed across a large number of processors, which must use small overhead space between them~\cite{axtmann2017place,parallel-inplace-radix, kuszmaul2020cache, gu2021parallel, axtmann2022engineering, penschuck2023engineering, hutton2025encoding} --- in this setting, ``in-place'' can mean either $\polylog(n)$ or $n^{1-\eps}$ space overhead.

\textbf{Fast Fourier transforms.} 
Another problem which has seen interest in time-efficient in-place computing is the Fourier transform.
The Cooley-Tukey fast Fourier transform algorithm can be implemented in-place for inputs of power-of-two lengths~\cite{burrus1981fft,temperton1991fft} --- further work has also shown that the ``truncated'' FFT can be implemented in-place for other input lengths~\cite{harvey2010place,arnold2013new,coxon2022place}, and found algorithms with better performance on certain hardware~\cite{hegland1994self, de2023place}. 

\textbf{Data compression.} 
A natural problem to want to compute in-place is data compression and decompression.
Existing work here is often driven by practical applications.
There have been papers discussing in-place ``delta compression'' for version control~\cite{burns2003place, klein2008modeling}, prefix codes which can be encoded and decoded in-place~\cite{milidiu1998place,chan2014selection}, and in-place implementations of the Burrows-Wheeler transform~\cite{crochemore2015computing, koppl2020place}.

\textbf{Computational geometry.}
Some problems that lend themselves to in-place computation take the form ``permute a list of points so that the new ordering encodes some geometric structure''.
There have been several such problems considered --- for instance, the problem of re-ordering points so that a prefix of them forms a convex hull~\cite{bronnimann2002place,bronnimann2004space,bronnimann2004towards,chan2010optimal}, so that a prefix forms the ``skyline'' (i.e. points which are not dominated in every coordinate by any other point)~\cite{blunck2010place}, or so that the permuted list can serve as an efficient nearest-neighbours data structure~\cite{bronnimann2004towards,chan2008place}.
Here, again, time-space tradeoffs are a major concern.

\subsubsection{Catalytic computing}

The model of catalytic space was introduced by Buhrman, Cleve, Kouck\'y, Loff, and Speelman~\cite{buhrman2014computing} as follows:

\begin{definition}
    A language $L$ belongs to $\CL$ if there exists a uniform algorithm with an $n$-bit read-only input tape, a $O(\log n)$-bit read/write worktape, and a $\poly(n)$-bit read/write catalytic tape such that, on every $x$, $\tau$, if the input tape is initialized to $x$ and the catalytic tape is initialized to $\tau$, then the algorithm terminates with an output of $L(x)$, and the catalytic tape still in configuration $\tau$.
\end{definition}

Buhrman, Cleve, Kouck\'y, Loff, and Speelman gave several structural results for this model: they showed that $\CL$ contains (uniform) $\TC^1$ and is contained in $\ZPP$, and gave oracles relative to which $\CL^O = \LOG^O$ and $\CL^O = \PSPACE^O$, respectively.
Since then, there has been substantial interest in understanding this class --- recent papers have found catalytic algorithms for problems not known to be in $\TC^1$~\cite{agarwala2025bipartite,alekseev2025catalytic}, collapsed the randomized and nondeterministic versions of $\CL$ to $\CL$~\cite{buhrman2018catalytic,datta2020randomized, cook2025structure, koucky2025collapsing}, and characterized how the class changes under weaker resetting requirements~\cite{bisoyi2024almost,gupta2024lossy,folkertsma2024fully}.
There has also been interest in variants of the catalytic model in other settings, such as branching programs~\cite{girard2015nonuniform,potechin2016note,cook2022trading,cote2023catalytic}, communication protocols~\cite{pyne2025catalytic}, and quantum computing~\cite{buhrman2025quantum}.
The underlying techniques behind this work has led to further applications in space-bounded derandomization~\cite{li2024distinguishing,pyne2024derandomizing,doron2024opening} and understanding the relationship between time and space~\cite{cook2020catalytic,cook2021encodings,cook2024tree,williams2025simulating}.
Nevertheless, many questions about catalytic computing remain poorly understood. 
The curious reader is referred to surveys of Kouck\'y and Mertz for further background~\cite{koucky2016catalytic,mertz2023reusing}.

\subsection{Roadmap}
In \cref{sec:summaries} we give sketches of the main proofs. In \cref{sec:defs} we formally define in-place computation and oracle in-place computation, and recall complexity background. In \cref{sec:FLsep} we separate $\FL$ and $\inplaceFL$. In \cref{sec:smallwidth} we show how to evaluate small-width circuits in $\inplaceFL$. In \cref{sec: fp_in_place} we give our oracle algorithm for computing $\FZPP$ in $\inplaceFL^O$. In \cref{sec:in_place_matrix} we give in-place algorithms for matrix problems. In \cref{sec:CLinP} we construct the oracle barrier to $\CL\subseteq \P$.

\section{Overview of proofs}\label{sec:summaries}
\subsection{Separating \ts{\( \FL \)}{FL} and \ts{\( \inplaceFL \)}{inplaceFL}}

To give a function (unconditionally) in $\inplaceFL$ but not $\FL$ we can simply take a function requiring linear space and apply it to the first 1\% of the input bits --- the remaining 99\% of the input is useless to a $\FL$ algorithm, but affords linear additional workspace to an $\inplaceFL$ algorithm.
The formal proof we give is similar but shows something slightly stronger, constructing a \emph{permutation} in $\inplaceFL \setminus \FL$. 
The same ideas work directly for $\inplaceFCL$ and $\FCL$, except that we do not know a uniform space hierarchy theorem for $\CL$, and so we require the additional assumption that $\CL \neq \PSPACE$.

We then show that, under cryptographic assumptions, $\FL$ is likewise not contained in $\inplaceFL$.
This follows from the observation that any permutation computable in $\inplaceFCL$ can be inverted in (randomized) polynomial time on average.
Thus, if there are one-way permutations computable in $\FL$, any such function is an example of something in $\FL \setminus \inplaceFCL$.
The basic idea behind the inversion is to note that an $\inplaceFCL$ machine using $c$ catalytic space has only $2^{n + c + O(\log n)}$ possible internal configurations, and that each of these is associated with at most one (input, initial catalytic setting) pair, since otherwise there would either be two distinct inputs that result in the same intermediate configuration (violating the assumption that we are computing a permutation) or two initial catalytic settings that result in the same intermediate configuration (violating the catalytic resetting requirement).
Thus, the number of intermediate configurations that eventually result in a given output and catalytic setting is $\poly(n)$ on average, meaning that if we see a random output and choose a random catalytic setting ourselves, we can likely invert in polynomial time by traversing backwards through the configuration graph. We note that this is the same fundamental idea that shows $\CL \subseteq \ZPP$ ~\cite{buhrman2014computing}.

\subsection{In-place algorithms for small width circuits}
We then show an $\inplaceFL$ algorithm for any function with a (logspace uniform) $\NC_2^0$ circuit.
The idea is to compute the output bits in a particular order that allows us to erase the input bits fast enough to store the results. 
As each output bit depends on only two input bits, we can think about such a circuit as an graph: each of the $n$ vertices corresponds to an input bit, and each of the $n$ edges corresponds to some output bit, computed as a function of its two endpoints.
We can safely erase an input bit once all edges incident to its corresponding vertex have been computed. 
One can see that, if we repeatedly take a vertex of minimum degree, compute the values of all incident edges, and then erase the value of that vertex, the total number of values we need to store at any point will be at most an additive constant larger than $n$.
However, it is not clear how to actually do the computational overhead of this process in low space --- the bulk of our proof consists of identifying a variant of this scheme that can be implemented using only logspace-computable properties of the original graph at each step.

Note that, given any (logspace uniform) list of functions that are each individually in $\inplaceFL$, we can compute their composition in $\inplaceFL$ by simply performing the transformations one-after-another.
So, this procedure also allows us to evaluate circuits of polynomial depth, as long as each layer has locality $2$, and the width is always bounded by $n + O(\log n)$. This gives us a $\CL$ algorithm solving $\lossycode$ for circuits of small width and depth: since $\CL$ can evaluate $\log$-depth circuits, we can check whether a chunk of our catalytic tape is successfully compressed and decompressed by the pair of circuits --- if not we have found a valid solution, and if so we can perform the compression in-place to free up additional workspace.

\subsection{Computing \ts{\( \FZPP \)}{FZPP} in-place with an oracle}
We next proceed to give oracle results demonstrating that an unconditional proof of $\FL \not\subseteq \inplaceFL$ is out-of-reach of current complexity theoretic techniques. Specifically, we construct a language $O \in \PH$ such that $(\FL \subseteq \FCL \subseteq) \FZPP \subseteq \inplaceFL^O$.
If $\SAT \in \LOG$, then the polynomial hierarchy collapses to $\LOG$ so an $\inplaceFL$ algorithm could simulate the oracle itself\footnote{We remark that, as with space-bounded computation in general, defining oracles for in-place computation requires care. In particular, there are some models of oracle access for which it would not be obvious that $O \in \LOG$ implies $\inplaceFL = \inplaceFL^O$. The oracle model we work with, defined in \cref{sec:defs}, is such that this collapse property is immediate.}, meaning that we would have $\FL \subseteq \inplaceFL$.

The fact that there exists \emph{any} decision oracle $O$ such that $\FZPP \subseteq \inplaceFL^O$ is not a priori obvious.
The issue is that, although the oracle may be incredibly powerful, it can only return a single bit at a time, so it must not only be powerful enough to compute any bit of the output function, but also able to help an $\inplaceFL$ machine mutate the input to the output one-step-at-a-time without destroying crucial information.
The approach we use to achieve this comes from a connection to existing literature on network routing.
In a well-studied scenario, 
a collection of users each want to transfer a message to some other user,
with the constraint that users correspond to the vertices of a $d$-dimensional hypercube and messages may only be passed along edges.
It is a classic result that this is possible with only $\poly(d)$ congestion --- i.e. only $\poly(d)$ messages must pass through any given vertex or edge~\cite{vocking2001almost}. 

Although it originated from the goal of managing traffic in communication networks, this result directly translates into a strategy for an oracle $O$ in our model.
Observe that we can think of the oracle as actually returning $O(\log n)$ bits as opposed to $1$, as the algorithm can specify which index into these $O(\log n)$ bits it wants with each call and store the results in its workspace.
Consider a routing strategy that lets each vertex $x$ send a message to the vertex $f(x)$ at $\poly(n)$ congestion.
The oracle will take as input an $n$-bit string (which it will think of as the name of a vertex on the hypercube), and an additional $O(\log n)$-bit tag (which it will think of as specifying an index into the list of all messages that ever pass through that vertex).
From this, it can uniquely determine the message, which allows it to determine these values for the message one timestep later in the routing.
It then provides as a response to the algorithm the new tag and the index of the bit to flip in the vertex name.

The above argument yields \emph{some} oracle with $\FZPP \subseteq \inplaceFL^O$ --- \cref{sec: fp_in_place} is dedicated to showing how to adapt these ideas so as to implement such an oracle in low complexity.
We show that this routing can be done efficiently if we work in a random choice of basis.
Then, we use tools from derandomization (including Li's beautiful $\FSTP$ algorithm for the range avoidance problem~\cite{li2024symmetric}) to make the oracle description constructive.

\subsection{Computing linear transformations in-place}
We next show new in-place algorithms for a basic linear algebraic problem: we can replace a vector $x$ with $Ax$ in-place, given read-only access to some matrix $A$.
Note that, as recently observed by Dumas and Grenet~\cite{dumas2024place}, if $A$ is upper-triangular, the \( i \)-th coordinate of \( Ax \) depends only on \( x_i, \dotsc, x_n \), so we can replace coordinates of \( x \) with \( Ax \) one-at-a-time in order.
A similar technique works if \( A \) is only \emph{almost upper-triangular}, meaning one additional diagonal is permitted to be non-zero.
So, in order to give an algorithm for general matrices, it suffices to efficiently work in a new basis under which $A$ is almost upper-triangular.

We show how to find such a basis in catalytic logspace.
We start with an arbitrary vector \( e \) and append \( e, A e, A^2 e, \dotsc \) to the basis until there is a linear dependence; we then start the process again with new vectors \( e' \) until the basis is complete.
Given an index $i$, we observe that it is possible to compute the $i$th basis vector in this list in logarithmic space given the ability to compute matrix ranks, which by \cref{lem:LinAlg} can be done in $\FCL$.
This new basis makes \( A \) almost upper-triangular by construction.
Using this procedure for matrix-vector multiplication, the $\FCL$ algorithm for matrix inversion, and a compress-or-random argument showing that we can reversibly manipulate the catalytic tape to contain a large invertible matrix, we are then also able to perform matrix inversion in $\inplaceFCL$.

\subsection{A barrier to \ts{$\CL\subseteq \P$}{CL in P}}
Our final result is an oracle $O$ such that $\CL^O = \EXP^O$. 
This improves upon the construction by Buhrman, Cleve, Kouck\'y, Loff, and Speelman of an oracle such that $\CL^O = \PSPACE^O$ \cite{buhrman2014computing} --- one consequence is that, by the time hierarchy theorem, our result rules out a relativizing proof of $\CL \subseteq \P$.
Buhrman et al.'s oracle construction follows a compress-or-random approach: given a string with very high Kolmogorov complexity (a ``password''), the oracle provides useful assistance in solving very hard problems, and given a string with low Kolmogorov complexity the oracle helps the algorithm compress, and later decompress, the string.
Even with the compression oracle's help, a $\PSPACE$ machine cannot produce a string with higher Kolmogorov complexity than its space bound, since its memory configuration after the final query before printing that string provides a succinct representation --- thus, if the requirement for Kolmogorov complexity is high enough, the algorithm will never find a password.
However, a $\CL$ machine with sufficient catalytic space can find a password: it simply runs the oracle on its tape, either finding a password immediately or compressing the tape contents until it can find a password by brute-force.
Note that it is crucial for this scheme that the compression and decompression be implementable in-place.
This is achieved by appealing to the chain rule for Kolmogorov complexity: if a string is very compressible, there must be some $O(\log(n))$-sized chunk that can be compressed conditional on the rest of the string --- this compression can be done in-place because the algorithm can store both the initial and final versions of that chunk simultaneously.

This approach is insufficient for our purposes, as an exponential-time machine with no space bound may make use of a decompression oracle to generate queries of very high complexity.
Instead of compressing based on Kolmogorov complexity, our oracle makes use of our observed connections to network routing to let the algorithm perform some more general types of in-place modifications to its catalytic tape.
We once again designate some strings as ``passwords'' --- our goal is to design an oracle that takes in a string and provides a single step of in-place modification to perform, such that repeatedly applying these transformations to any initial configuration of a catalytic machine will eventually first pass through a password, and then later eventually reset to the initial configuration.
In order to design such an oracle that also prevents an exponential-time algorithm from being able to find a password, we define a combinatorial game played by a ``cycle-finder'' (representing a dovetailing enumeration of all exponential time oracle algorithms) that adaptively chooses some configurations to query, and a ``cycle-hider'' (representing our oracle construction) that responds with those configurations' successors.
We use ideas from randomized routing to describe a strategy for the cycle-hider to win this game, and show that any such strategy gives our desired oracle construction.

\section{Preliminaries}\label{sec:defs}
We assume familiarity with the basics of Turing machines~\cite{AB}. We write \( \log x \) to denote the base-2 logarithm \( \log_2 x \).
We let $U_n$ denote the uniform distribution on $\zo^n$. For a language $L$, we let $L(x)$ denote the indicator $\Ind{x\in L}$. 

\subsection{Space-bounded classes}
We recall the standard definition of space-bounded computable functions:
\begin{definition}
    A function family $\{f_n:\zo^n\ra \zo^{m(n)}\}_{n\in \N}$ is in $\FSPACE[s(n)]$ if there is a Turing machine $\cM$ using $s(n)$ bits of workspace such that, on any input $x$ of size $n$, $\cM(x)$ writes $f_n(x)$ to a write-only output tape.
\end{definition}

The model of catalytic space was introduced by Buhrman, Cleve, Kouck\'y, Loff, and Speelman~\cite{buhrman2014computing}; again we define the functional version here:

\begin{definition}
    A function family $\{f_n:\zo^n\ra \zo^{m(n)}\}_{n\in \N}$ is in $\FCSPACE[s,c]$ if there is a Turing machine $\cM$ using $s$ bits of workspace as well as $c$ additional bits of workspace, called the \textit{catalytic memory}, such that, on any input $x$ of size $n$ and any initial configuration $\tau$ of the catalytic memory, $\cM(x)$ writes $f_n(x)$ to a write-only output tape and halts with the catalytic memory in configuration $\tau$.
\end{definition}

\begin{definition}[generalization of \cite{chen2024symmetric}]
    A function family $\{f_n:\zo^n\ra \zo^{m(n)}\}_{n\in \N}$ is in $\FZPP$ if there exists a polynomial time randomized algorithm $\mathcal{A}$ such that on input $x \in \zo^n$, $\mathcal{A}(x)$ outputs either $f(x)$ or $\bot$ and the probability (over the internal randomness of $\mathcal{A}$) that $\mathcal{A}(x)$ outputs $f(x)$ is at least $2/3$.
\end{definition}

We now define in-place $\FSPACE$ and $\FCSPACE$ computation. Note that while we overwhelmingly work with length-preserving functions, we give a definition that holds for compressing and extending functions.
\begin{definition}[inplaceFSPACE]
    A function family $\{f_n:\zo^n\ra \zo^{m(n)}\}_{n\in \N}$ is in $\inplaceFSPACE[s]$ if there is a Turing machine $\cM$ that, when run on a tape in configuration $x\circ 0^{\max\{0,m(n)-n\}}\circ 0^s$ for $n = |x|$, halts with the tape in configuration $f_n(x)\circ 1^{\max\{0,n-m(n)\}}\circ 1^s$. 
\end{definition}

\begin{definition}[inplaceFCSPACE]
    A function family $\{f_n:\zo^n\ra \zo^{m(n)}\}_{n\in \N}$ is in $\inplaceFCSPACE[s,c]$ if there is a Turing machine $\cM$ that, when run on a tape in configuration $x\circ 0^{\max\{0,m(n)-n\}}\circ 0^s \circ \tau$ for $n = |x|$ and any initial configuration $\tau\in \zo^c$, halts in configuration $f_n(x)\circ 1^{\max\{0,n-m(n)\}}\circ 1^s \circ \tau$.
\end{definition}

Throughout this paper, we will almost entirely focus on the case of logspace and
catalytic logspace:

\begin{definition}
    We define
    \begin{itemize}
        \item $\FL := \cup_{k \in \N}$ $\FSPACE[k \log n]$
        \item $\FCL := \cup_{k \in \N}$ $\FCSPACE[k \log n, n^k]$
        \item $\inplaceFL := \cup_{k \in \N}$ $\inplaceFSPACE[k \log n]$
        \item $\inplaceFCL := \cup_{k \in \N}$ $\inplaceFCSPACE[k \log n, n^k]$
    \end{itemize}
\end{definition}

\subsubsection{Space-bounded oracle computation}
\label{subsec: oracle_defs}
The \say{correct} definitions of in-place (and catalytic) computation relative to oracles is non-obvious. 
First, we recall the most common definition of oracles in the standard space-bounded model:
\begin{definition}
    \label{def: standard_oracle_def}
    Let $O: \zo^* \rightarrow \zo$ be any function. An $\FL^O$ machine $\cM$ is
    defined as an $\FL$ machine with an additional write-only \emphdef{oracle tape} as
    well as a read-only \emphdef{oracle output bit} initialized to 0.
    In addition to its normal operations, $\cM$ may also perform an \emphdef{oracle query} by transitioning to a special state in the FSM,
    which reads the current contents $y$ of the oracle tape, overwrites
    $O(y)$ onto the oracle output bit; and then erases the oracle tape.
\end{definition}

For in-place computation (and indeed, for catalytic computation), this write-only query tape can be used to persist intermediate data (for instance, the algorithm could write the input $x$ to this query tape, modify the read-write tape into $x'$, then append $x'$ to the query tape and query $O(x,x')$). As such, we restrict our in-place algorithms to only querying the oracle on its existing read-write tape:
\begin{definition}
    \label{def: our_oracle_def}
    Let $O: \zo^* \rightarrow \zo$ be any function. An $\inplaceFL^O$ machine $\cM$ is
    defined as an $\inplaceFL$ machine with an additional read-only \emphdef{oracle output bit}
    initialized to 0. In addition to its normal operations, $\cM$
    may also perform an \emphdef{oracle query} on any subinterval $T$ of the tape,\footnote{We assume there is an oracle query state in the finite state machine, and the final $2\lceil \log (n+O(\log n)\rceil$ bits of the tape specify the interval of the query.}
    which reads the current contents $y$ of $T$ and overwrites $O(y)$ onto
    the oracle output bit.
\end{definition}
Note that as written, this restricts $\inplaceFL$ to making oracle queries of length $n+O(\log n)$. We view this as the most natural model of in-place oracles. In fact, our oracle relative to which $\CL^O=\EXP^O$ is also used by the catalytic algorithm in this manner (i.e. it is only invoked on a subsection of the input and catalytic tape).\footnote{A precise definition of oracles for $\CL$ has not been previously written down, but the oracle for which $\CL^O=\PSPACE^O$ of~\cite{buhrman2014computing} is also used in this fashion.}

The most important property our oracle obeys is that it plays nicely with the oracle itself being computable in logspace, which we require for our conditional separations:
\begin{proposition}[Oracle Collapse]\label{prop:oraclecollapse}
    Suppose $f\in \inplaceFL^{O}$ and $O\in \LOG$ (resp. $O\in \CL$). Then $f\in \inplaceFL$ (resp. $f\in \inplaceFCL$).
\end{proposition}
\begin{proof}
    For the first result, let $\cM$ be the oracle machine that computes $f$, and $\cA$ a logspace machine that computes $O$. Then our new machine $\cM'$ emulates $\cM$ step-by-step, and when an oracle call is made on some section $\pi$ of the tape, pauses the simulation and simulates $\cA(\pi)$ using an additional $O(\log n)$ bits of auxiliary space; after $\cA$ returns an answer, we continue running $\cM$. Our total space usage will be the space to run $\cM$, namely $n + O(\log n)$ bits, plus an additional $O(\log n)$ bits for $\cA$, which gives $n + O(\log n)$ bits in total.

    For the latter result, when an oracle call is made on some section $\pi$ of the tape, we pause the simulation and use an additional $O(\log n)$ bits of auxiliary space and the catalytic tape to simulate $\cA(\pi)$ (and note that by the definition of $\CL$ this call terminates with the catalytic tape unmodified).
\end{proof}

\subsection{Circuit classes}
We always work with the full basis $B_k$ for circuits of fan-in $k$.

We slightly abuse notation and identify circuit classes with their functional versions.
\begin{definition}
    We say $f\in \NC^i$ if there is $c>0$ and a family of circuits $\{C_n\}_{n\in \N}$ where $C_i$ has size $n^c$ and depth $c\log^i(n)$, and for every $x\in \zo^i$ we have $f(x)=C_i(x)$. We say a circuit family is \emphdef{logspace uniform} if there is an algorithm that on input $1^i$ runs in space $O(\log i)$ and outputs $C_i$. We let $\unifNC^i$ denote logspace-uniform $\NC^i$. 
\end{definition}
\begin{definition}
    We say $f \in \NC^0_\ell$ if each output bit of $f$ is a function of at most $\ell$ input bits. (Equivalently, $f$ is computable by a depth-$1$ circuit of fan-in $\ell$.)
\end{definition}

\newcommand{\cC}{\mathcal{C}}
For a class of circuits $\cC$, we let $\mathcal{C}\text{-eval}$ be the language $\{(C,x):C\in \cC, C(x)=1\}$.

\subsection{Compression classes}
\label{subsec: compression_classes}
The \emphdef{range avoidance} problem, introduced in \cite{kleinberg2021total}, can be viewed as an algorithmization of the union bound.
\begin{definition}[\cite{kleinberg2021total}]
    \label{def: avoid}
    Let $\AVOID$ be the search problem whose input is a circuit $C: \zo^n \rightarrow \zo^{n+1}$, and whose valid outputs are $y \in \zo^{n+1}$ such that for all $x \in \zo^n$, $C(x) \neq y$.
    
    We further say a procedure solves \emphdef{single-value $\AVOID$ ($\svAVOID$)} if for every avoid instance $C$, there is a single $y_C$ which is a solution to $\AVOID(C)$ such that the procedure always outputs $y_C$ on input $C$.
\end{definition}
$\AVOID$ asks us to find a $y$ which is avoided by the circuit $C$. Notice that such a $y$ is guaranteed to exist by the dual pigeonhole principle. Many explicit construction problems (e.g. the construction of hard truth tables) reduce to $\AVOID$ \cite{korten2022hardest}. We will use an $\AVOID$ oracle in \cref{sec: fp_in_place} when we find that we need certain advice  strings which are guaranteed to exist by the probabilistic method, but are nontrivial to construct explicitly. We will also crucially rely on the fact that $\AVOID$ (and in fact $\svAVOID$) is in the class $\FSTP$, and consequently can be solved by its decision version ($\STP$).

\begin{definition}[\cite{chen2024symmetric}]
    A single valued $\FSTP$ algorithm $A$ is specified by a polynomial $\ell(\cdot)$ together with a polynomial time verification algorithm $V_A(x, \pi_1, \pi_2)$. On an input $x \in \zo^*$, we say $A$ outputs $y_x$, if the following hold:
    \begin{enumerate}
        \item There is a $\pi_1 \in \zo^{\ell(|x|)}$ such that for every $\pi_2 \in \zo^{\ell(|x|)}$, $V_A(x, \pi_1, \pi_2)$ outputs $y_x$.

        \item There is a $\pi_2 \in \zo^{\ell(|x|)}$ such that for every $\pi_1 \in \zo^{\ell(|x|)}$, $V_A(x, \pi_1, \pi_2)$ outputs $y_x$.
    \end{enumerate}
    And we say that $A$ solves a search problem $\Pi$ if on input $x$ it outputs a string $y_x$ and $y_x \in \Pi_x$ (where $\Pi_x$ is the set of possible solutions on input $x$).
\end{definition}

\begin{definition}[\cite{russell1998symmetric}]
    A language $\mathcal{L} \subseteq \zo^*$ is in $\STP$ if there exists a polynomial time verifier $V$ such that the following holds.
    \begin{enumerate}
        \item If $x \in \mathcal{L}$, there exists a $y$ such that for all $z$, $V(x, y, z) = 1$;
        \item If $x \notin \mathcal{L}$, there exists a $z$ such that for all $y$, $V(x, y, z) = 0$.
    \end{enumerate}
\end{definition}

Observe that $\STP$ is contained in $\mathsf{\Sigma_2^P} \cap \mathsf{\Pi_2^P}$. We will rely on the following breakthrough result to connect our use of an $\AVOID$ oracle to $\FSTP$.
\begin{lemma}[\cite{li2024symmetric}]
    There is a single-valued $\FSTP$ algorithm solving $\svAVOID$.
\end{lemma}

\begin{corollary}
    \label{cor: svAvoid_using_S2P}
    There is an $\FP^{\STP}$ algorithm which solves $\svAVOID$.
\end{corollary}
\begin{proof}
    Let $V_A(x, \pi_1, \pi_2)$ be the verifier that confirms $\svAVOID$ is in $\FSTP$. Note that if we structure $\svAVOID$ as a decision problem, where on input $(C,i)$ the output is the $i$th bit of the canonical solution $y_C$, then this problem is in $\STP$. The verifier $V'((C, i), \pi_1, \pi_2)$ that proves this simply outputs whether bit $i$ of $V(C, \pi_1, \pi_2) = y_C$ is $1$. Therefore, there is an $\FP^{\STP}$ algorithm which on input $C$ outputs $y_C$.
\end{proof}

On occasion, it will be useful to think of the input and output of a circuit $C: [A] \rightarrow [B]$ we feed to an $\AVOID$ oracle as integers rather than bitstrings. We now argue that this is almost without loss of generality. We can always take a circuit $C: [A] \rightarrow [B]$ and turn it into a circuit which works over bits $C': \zo^{\lceil \log_2(A) \rceil} \rightarrow \zo^{\lfloor \log_2(B) \rfloor}$ by defining $C'(x) = \text{bitstring}(C(\text{int}(x)))$ for all $x$ where $\text{int}(x) < A$ and $C(\text{int}(x)) < 2^{\lfloor \log_2(B) \rfloor}$, and $C'(x) = 0^{\lfloor \log_2(B) \rfloor}$ otherwise (here $\text{bitstring}(\cdot)$ is the canonical map from integers to bitstrings and $\text{int}(\cdot)$ is its inverse). $C'$ may be slightly less expanding than $C$, but if $B/A \geq 8$, then $C'$ still expands by at least 1 bit and will therefore be a valid input to $\AVOID$. This also works for $C: [A_1] \times \dots \times [A_t] \rightarrow [B_1] \times \dots \times [B_t]$ if we simply reinterpret $C$ as $C: [A_1 \times \dots \times A_t] \rightarrow [B_1 \times \dots \times B_t]$.

Finally, we define the \emphdef{lossy coding} problem, which captures explicit constructions where we have explicit compression \textit{and} decompression algorithms:
\begin{definition}[\cite{korten2022derandomization}]
    For a circuit class $\mathcal{C}$, let $\mathcal{C}\text{-}\lossycode$ be the
    search problem whose input is a pair of circuits $C, D\in \mathcal{C}$ with
    $C \colon \bin^n \to \bin^{n-1}$ and $D \colon \bin^{n-1}\to\bin^n$, and whose valid outputs are $\{x \ | \ D(C(x)) \neq x\}$.
\end{definition}

\subsection{Cryptography}
Assuming cryptographic primitives can be useful for proving conditional lower bounds, of which we focus on one-way permutations:
\begin{definition}\label{def:OWP}
    We say a length-preserving function $f:\zo^*\ra\zo^*$ is a (uniform) \emphdef{one-way permutation} if for every $n\in \N$, $f$ restricted to $n$-bit inputs is a permutation, and for every randomized polynomial-time algorithm $\cA$,
    \[
    \Pr_{x}[\cA(x)=f^{-1}(x)] \le n^{-\omega(1)}.
    \]
\end{definition}

For us, the relevant machines for attempting to break one-way permutations are \textit{average-case} complexity classes.

\begin{definition}[avgP and avgZPP]
    We say $L$ is computable in (errorless) $\avgP_\eps$ (resp. $\avgZPP_\eps$) if there is a polynomial time (resp. randomized polynomial time) machine $\cM$ such that for every $n\in \N$, $\cM(x)\in \{L(x),\perp\}$ and $\cM$ outputs $\perp$ with probability at most $\eps$ over $x\la U_n$ (resp. over $x\la U_n$ and the internal randomness of $\cM$). 
\end{definition}

\subsection{Finite fields and linear algebra}\label{sec:algebra}

A number of our techniques use linear algebra over finite fields.
We will thus show, as prerequisites, how to store and manipulate field elements efficiently, and review past work which implies that matrix rank and inverse, and the product of a list of matrices, can be computed in \( \FCL \).

\begin{definition}\label{def:RepresentableField}
A finite field \( K \) is \emph{representable in space \( s \)} if there is an injective function \( r : K \to \zo^s \) and algorithms \( \mathsf{ADD} \), \( \mathsf{MULTIPLY} \) and \( \mathsf{VALID} \) which in space \( O(s)  \) compute the following.
\( \mathsf{ADD} \) and \( \mathsf{MULTIPLY} \) on input \( r(x), r(y) \) compute \( r(x+y), r( x \cdot y ) \) respectively.
\( \mathsf{VALID} \) determines whether its input is in the range of \( r \).
When $s = O(\log(n))$, we simply say $K$ is \emph{representable}.
If \( r \) is a bijection, we say \( K \) is \emph{exactly representable}.
\end{definition}
Note that \( \mathsf{MULTIPLY} \) and \( \mathsf{VALID} \) are enough to compute inverses by enumerating all possibilities.

For prime \( p \) and positive integer \( k \), we write \( \GF( p^k ) \) for the unique (up to field automorphisms) finite field with \( p^k \) elements.
\begin{lemma}\label{lem:RepresentableFields}
Every finite field \( \GF( p^k ) \) is representable in space \( k \lceil\log p \rceil \), and \( \GF( 2^k ) \) is exactly representable in space \( k \).

These representations are uniform in the sense that there are universal algorithms \( \mathsf{ZERO}^* \), \( \mathsf{ONE}^* \), \( \mathsf{ADD}^* \), \( \mathsf{MULT}^* \), \( \mathsf{VALID}^* \) which take \( p \) and \( k \) as parameters and produce the representations of \( 0 \) and \( 1 \) in \( \GF( p^k ) \), and implement \( \mathsf{ADD} \), \( \mathsf{MULTIPLY} \), and \( \mathsf{VALID} \) \( \GF( p^k ) \), respectively.
\end{lemma}
\begin{proof}
    Denote the ring of polynomials with coefficients in \( \GF(p) \) by \( \GF(p)[x] \).
    For any irreducible polynomial \( q(x) \in \GF(p)[x] \) of degree \( k \), the field \( \GF( p^k ) \) is isomorphic to \( \GF(p)[x] / (q(x)) \).
    In particular, we can identify the elements of \( \GF( p^k ) \) with polynomials with coefficients in \( \GF(p) \) and degree less than \( k \).
    A polynomial can be stored as a list of \( k \) integers from \( 0 \) to \( p-1 \), each occupying \( \lceil \log p \rceil \) bits.
    If \( p=2 \), coefficients are in \( \zo \), so the representation is exact.
    To add field elements, add their corresponding polynomials, and to multiply them, first multiply their polynomials and then reduce modulo \( q(x) \) to get a polynomial of degree less than \( k \).
    \( \mathsf{ZERO}^* \) and \( \mathsf{ONE}^* \) produce the polynomials \( 0 \) and \( 1 \), respectively.
    The \( \mathsf{VALID}^* \) function checks that each coefficient is between \( 0 \) and \( p-1 \) (which will always be true if \( p=2 \)).
    All of this is straightforward to compute in space \( k^{ O(1) } \) once the polynomial \( q(x) \) is determined.
    
    In order for the representations to be well-defined, and to implement the universal algorithms \( \mathsf{ADD}^* \) and \( \mathsf{MULTIPLY}^* \), we need a consistent choice of irreducible polynomial \( q_k(x) \) for each \( k \).
    
    Define \( q_k(x) \) to be the lexicographically first irreducible polynomial of degree \( k \).
    This can be found in space \( k^{ O(1) } \) by trying polynomials until one is determined to be irreducible.
    (To test whether a polynomial is irreducible, test all possible factors.)
\end{proof}

\begin{lemma}[Matrix operations]\label{lem:LinAlg}
For any representable field \( K \), the following can be computed in \( \FCL \):
\begin{itemize}
\item Multiplying many matrices: compute \( \prod_{ i=1 }^{ \ell } A_i \) given \( \ell \) matrices \( A_i \in K^{ n \times n } \).
\item The rank of a matrix \( A \in K^{ m \times n } \).
\item The inverse \( A^{ -1 } \) of an invertible matrix \( A \in K^{ n \times n } \).
\end{itemize}
\end{lemma}

\begin{proof}
Past work proves the result for integer matrices, and by extension, fields with a prime number of elements, and in general reduces the second and third problems to the first problem over any ring with unity.
All that remains is to complete the proof for arbitrary finite fields.

The product \( \prod_{ i=1 }^{ \ell } A_i \) can be computed in \( \FCL \) using the technique of Ben-Or and Cleve \cite{BenorCleve92}, which computes a depth-\( d \) algebraic formula over any ring (in this case, a ring of matrices) using \( O( 4^d ) \) invertible operations.
Alternatively, Lemma~4 (applied recursively) or Lemma~8 of Buhrman, Cleve, Kouck\'y, Loff and Speelman \cite{buhrman2014computing} will accomplish the same thing.

From here, Allender, Beals and Ogihara~\cite{rankEtc} serve as a useful guide through past work.
Computing the determinant \( \Det( M ) \) of any \( M \in K^{ n \times n } \) reduces to computing \( \prod_{ i=1 }^{ \ell } A_i \) via their Proposition~2.2, which they attribute to Valiant~\cite{ValiantBoolDifficult} and \cite{TodaDet}.

This immediately allows us to compute matrix inverses, since for any matrix $Q$ the \( (i,j) \)-th entry of \( Q^{ -1 } \) equals $ (-1)^{ i+j } \Det( Q_{ -j,-i } ) / \Det( Q ) $, where \( Q_{ -j,-i } \in K^{ (n-1) \times (n-1) } \) is \( Q \) with the \( j \)-th row and \( i \)-th column deleted.

Finally, Mulmuley~\cite{mulmuleyRank} and von zur Gathen~\cite{gathenParallelLinAlg} reduce computing the rank of a matrix to computing \( \prod_{ i=1 }^{ \ell } A_i \).
\end{proof}

We also mention a folklore result that two blocks of memory can be efficiently
swapped with no additional memory. This will be useful in a number of
proofs.
\begin{lemma}\label{lem:swap}
    Let $\langle a,b \rangle \in \{0,1\}^m \times \{0,1\}^m$ for any $m \in \mathbb{N}$.
    Then SWAP$(\langle a,b \rangle) = \langle b,a \rangle$ can be computed in $\inplaceFL$.
\end{lemma}
\begin{proof}
    Letting $R_a$ and $R_b$ be the memory holding $a$ and $b$ respectively, and letting
    $\oplus$ denote coordinate-wise XOR, the following instructions compute SWAP:
    \begin{enumerate}
        \item $R_a = R_a \oplus R_b$
        \item $R_b = R_a \oplus R_b$
        \item $R_a = R_a \oplus R_b.$\qedhere
    \end{enumerate}
\end{proof}

\section{Separating \ts{$\FL$}{FL} and \ts{$\inplaceFL$}{inplaceFL}}\label{sec:FLsep}
\subsection{Permutations computable in-place but not with transducers}
We first show that there is a function in $\inplaceFL$ that does not lie in $\FL$. 
\begin{theorem}\label{thm:flsep}
    There is a permutation $f\in \inplaceFL$ such that $f\notin \FSPACE[n/\omega(1)]$.
\end{theorem}
\begin{proof}
    Let $L_{hard}$ be a unary language in $\SPACE[n]\setminus \SPACE[n/\omega(1)]$, which exists per the space hierarchy theorem~\cite{stearns1965hierarchies}. Let 
    \[
    f(x) = \begin{cases}
        x & x\neq 0^{n-1}b \text{ for } b \in \bin\\
        0^{n-1}b & x=0^{n-1}b\text{ and } 0^n \not\in L_{hard}\\
        0^{n-1}\overline{b} & x=0^{n-1}b\text{ and } 0^n \in L_{hard}.
    \end{cases}
    \]
    We first claim $f\in \inplaceFL$. The $\inplaceFL$ algorithm for $f$ leaves the input intact unless all but the last bit are $0$s. Otherwise, the algorithm uses the $O(\log n)$ bits of extra space to remember the length of $x$, erases the first $n - 1$ bits, computes $L_{hard}$ on the specified input length, and either flips or preserves the final bit as appropriate. 
    
    We now claim that $f\notin \FSPACE[n/\omega(1)]$. Note that for $g\in \FSPACE[s]$ the language $L_{first}=\{x:g(x)_n=1\}$ is in $\SPACE[s]$, and hence if $f\in \FSPACE[n/\omega(1)]$ we would have $L_{hard}\in \SPACE[n/\omega(1)]$, violating the space hierarchy theorem.
\end{proof}

This immediately implies \cref{prop:inplacehax}, and the proof of \cref{prop:inplacehaxcl} is the same:

\begin{proof}[Proof of \cref{prop:inplacehaxcl}]
    Let $L_{hard}$ be a $\SPACE[n]$-complete language under logspace reductions. If $L_{hard}\in \CL$, we would have $\PSPACE=\CL$ via a simple padding argument, so this does not occur by assumption.
    We define $f$ as in \cref{thm:flsep}; once again this function is in $\inplaceFL$ (and hence $\inplaceFCL$), but if $f \in \FCL$ then $L_{hard} \in \CL$.
\end{proof}

In fact, because the space hierarchy theorem relativizes, both of these results hold relative to every oracle $O$ as well.

\subsection{Average-case inversion of permutations in \ts{$\inplaceFL$}{inplaceFL}}\label{sec:inversion}
It will be convenient to define configuration graphs of in-place FL machines:
\begin{definition}[inplaceFL configuration graph]
    Given a machine $\cM$ computing $f$ in $\inplaceFL$, for every $n\in \N$ we define the \emphdef{configuration graph $\cG_{\cM}$} of $\cM$ to be the graph with vertices in $(\tau,\mu)$, where $\tau\in \zo^n$ holds the $n$ bits of the in-place tape and $\mu\in \zo^{O(\log n)}$ holds all other bits of the configuration. We assume WLOG that the machine starts in configuration $(x,\vz)$ and halts in configuration $(f(x),\vo)$ for every input $x$. We let $\Gamma^{-1}[(\tau,\mu)]$ be the set of configurations that reach $(\tau,\mu)$.
\end{definition}
We can easily extend this definition to the configuration graph of catalytic machines:
\begin{definition}[inplaceFCL configuration graph]
    Given a machine $\cM$ computing $f$ in $\inplaceFCL$, for every $n\in \N$ we define the \emphdef{configuration graph $\cG_{\cM}$} of $\cM$ to be the graph with vertices in $(\tau,\mu,w)$, where $\tau\in \zo^n$ holds the $n$ bits of the in-place tape, $w$ holds the catalytic tape, and $\mu\in \zo^{O(\log n)}$ holds all other bits of the configuration. We assume without loss of generality that the machine starts in configuration $(x,\vz,w)$ and halts in configuration $(f(x),\vo,w)$ for every input $x$ and initial tape $w$. We let $\Gamma^{-1}[(\tau,\mu,w)]$ be the set of configurations that reach $(\tau,\mu,w)$.
\end{definition}
Moreover, observe that this configuration graph has out-degree one (except for halt states which have out-degree zero), $2^n$ start and halting configurations, and $2^n\cdot \poly(n)$ total vertices. We note that in- and out-configurations can be computed easily.

\begin{fact}\label{fct:enum}
    There is a logspace algorithm that, given $(\tau,\mu)$ (resp. $(\tau,\mu,w)$), enumerates all neighboring (in and out) configurations.
\end{fact}

We can then prove the main result:\footnote{All of these results relativize, where the inversion is then in $\avgP_{n^{-c}}^O$.}
\begin{lemma}
    Let $f\in \inplaceFL$ be a family of permutations. Then $g=f^{-1}$ can be computed in $\avgP_{n^{-c}}$ for every $c$.
\end{lemma}
\begin{proof}
    Let $\cM$ be the $O$-oracle machine that computes $f$ and let $\cG_{\cM}$ be the configuration graph of this machine. Since there are $2^n$ halting configurations $(y,\vo)$ and \( 2^n \cdot n^{ O(1) } \) total configurations $(\tau,\mu)$, and since every configuration reaches at most one halting configuration, we can conclude that with probability at least $1-n^{-c}$ the size of $\Gamma^{-1}[(y,\vo)]$ is $n^{c+O(1)}$.
    
    Then our inversion algorithm, on input $y$, computes the component of configurations reaching $(y,\vo)$ using \cref{fct:enum}. If we discover more than $n^{c + O(1)}$ reachable states, we halt and fail to invert. Otherwise we will exhaust the entire component. Since $f$ is a permutation, and $\cM$ successfully computes it on every input, there must be exactly one start vertex $(x, \vz)$ in this component and we must have $f(x) = y$. So, as long as the component has size at most $n^{c+O(1)}$ we will find an inverse.
\end{proof}

The result extends immediately to $\inplaceFCL$, at the cost of the inversion algorithm becoming randomized as we must draw a random catalytic tape.
\begin{lemma}\label{lem:ipFCLinvert}
    Let $f\in \inplaceFCL$ be a family of permutations. Then $g=f^{-1}$ can be computed in $\avgZPP_{n^{-c}}$ for every $c$.
\end{lemma}
\begin{proof}
    Let $\cM$ be the catalytic machine that computes $f$, let $\cG_{\cM}$ be the configuration graph of this machine, and let $s=\poly(n)$ be the size of the catalytic tape. As there are $2^n\cdot 2^s$ states of the form $(y,\vo,w)$ and each such state is a halt state, the size of $\Gamma^{-1}[(y,\vz,w)]$ is at most $m=n^{c+O(1)}$ with probability at least $1-n^{-c}$ over a random $y$ and $w$.  
    
    Our algorithm, given $y$, draws a random $w$ and enumerates elements of $\Gamma^{-1}[(y,\vo,w)]$ using \cref{fct:enum} until we either find an element of the form $(x,\vz,w)$ in which case we halt and return $x$ (note that $f(x)=y$ by the correctness of $\cM$) or find at least $m$ elements, where we abort and return $\perp$. It is clear that this inversion algorithm succeeds with the claimed probability. 
\end{proof}

Using this result, we establish that natural cryptographic conjectures imply there is $f\in \FL\setminus \inplaceFL$:
\begin{theorem}
    Assume there is a one-way permutation $f: \zo^n \rightarrow \zo^n$ computable in logspace-uniform $\TC^1$. Then $\FCL\not\subseteq \inplaceFCL$.
\end{theorem}
\begin{proof}
    By~\cite{buhrman2014computing}, $f$ is computable in $\FCL$. However, by the assumption that $f$ is a OWP we have $f\notin \avgZPP_{1/n}$, so by \cref{lem:ipFCLinvert} we have $f\notin \inplaceFCL$.
\end{proof}
We recall a result of \cite{applebaum2006cryptography}:
\begin{theorem}[Theorem 5.4 \cite{applebaum2006cryptography}]\label{thm:aik}
    Suppose there is a OWP computable in $\FL$. Then there is a OWP computable in $\unifNC^0_4$.
\end{theorem}
We remark that their result does not state the latter OWP is in logspace uniform, however this can be seen from the construction.\footnote{Given a (deterministic, read-many) set of branching programs $B_1,\ldots,B_n$ such that $x\ra (B_1(x),\ldots,B_n(x))$ computes an OWP, their result constructs a degree-$3$ randomized encoding of each $B_i$, then computes a local encoding of each output bit. Both constructions are clearly computable in logspace.}

Thus, we obtain that the existence of cryptography implies hardness of evaluating $\NC^0_4$ functions and integer multiplication:
\owphard*
\begin{proof}
    By \cref{thm:aik}, the existence of a one-way permutation computable in $\FL$ implies the existence of a one-way permutation computable in $\unifNC^0_4$. However, by \cref{lem:ipFCLinvert} we know that $\inplaceFCL$ cannot compute any one-way permutation.
\end{proof}   

We note one specific corollary in the case of multiplication:
\begin{corollary}\label{cor:rsa}
    Let $\mathsf{Mult}\colon \bin^* \ra\bin^*$ be the length-preserving function that takes a concatenated pair of $n$-bit integers $x$ and $y$, and returns their $2n$-bit product $x * y$. Assuming that the RSA cryptosystem is secure, $\mathsf{Mult} \not\in \inplaceFCL$.
\end{corollary}
\begin{proof}
    We can solve factoring in the average-case by inverting $\mathsf{Mult}$. We cannot directly apply \cref{lem:ipFCLinvert}, as $\mathsf{Mult}$ is not a permutation. However, it is close enough to a permutation for our purposes. In order to break RSA, it suffices to give an expected polynomial time algorithm to factor $pq$ for $p$ and $q$ each uniform random primes between $2^{n-1}$ and $2^n$. By the prime number theorem, we know that there are $\Omega(2^n/n)$ primes between $2^{n-1}$ and $2^n$, and so there are $\Omega(2^{2n}/n^2)$ possible RSA semiprimes $pq$. As before, each of these outputs must belong to disjoint components of the configuration graph, so for a random $pq$ and a random setting of the catalytic tape the corresponding component of the configuration graph will have size at most $\big(2^{2n + O(\log n)}\big)\big({\Omega(2^{2n}/n^2)}\big)^{-1} = n^{O(1)}$ in expectation. We can thus find $p$ and $q$ in expected polynomial time by traversing the entire component until we find an input state.
\end{proof}

\section{In-place algorithms for small-width circuits}\label{sec:smallwidth}

In this section we show in-place algorithms for evaluating restricted circuit
classes. To see where the circuit structure may play a role in in-place computation,
we observe the following easy simulation:

\begin{proposition}\label{prop:only-depend-on-earlier-computable}
    For any $f \in \unifNC^1$, if $f$ is not length-extending, and, for all $i$, $f$'s $i$th output bit depends only on its first $i+O(\log n)$ input bits, then $f \in \inplaceFL$.
\end{proposition}
\begin{proof}
    Because $\unifNC^1 \subseteq \LOG$, we can compute any given bit of the output in logspace. We will compute the output bits in order from last to first. Once we've computed $k$ bits of the output, we only need to remember the first $n + O(\log n) - k$ bits of the input, so we can erase the end of the input as we go, and thus we need only keep $n + O(\log n)$ bits in memory at all times.
\end{proof}

For the rest of this section we will study restricted width circuits. These are a natural class of circuits for $\inplaceFL$ since one can hope that the natural strategy of computing the circuit layer-by-layer in-place might work. Note that it is unclear if such a layer-by-layer evaluation strategy works in the case that the circuit $C: \zo^n \rightarrow \zo^n$ we wish to evaluate has width $n + \omega(\log n)$ since our program will not have enough space to write down the outputs at $C$'s largest layer. Therefore, we only concern ourselves with circuits of width $n + O(\log n)$ and show that such circuits of fan-in $2$ can indeed be computed in $\inplaceFL$.
\begin{theorem}\label{thm:small-width-computable}
    For any $f\colon \bin^*\to\bin^*$ with logspace uniform, width-$(n + O(\log n))$, fan-in $2$ circuits, $f \in \inplaceFL$.
\end{theorem}

We use this result to show that $\LOSSY$ for $O(\log n)$ depth, small-width circuits lies in $\CL$.
{\renewcommand{\footnote}[1]{}
\lossysmallwidth*
}
\begin{proof}[Proof assuming \cref{thm:small-width-computable}]
    Recall that the input to the problem is $(C,D)$, where $C:\zo^n\ra\zo^{n-1}$ and $D:\zo^{n-1}\ra\zo^n$ are circuits of depth $O(\log n)$, fan-in-$2$, and width $n+O(\log n)$. 

    The catalytic algorithm works as follows. Let $\tau$ be the initial tape, and divide the tape into $(\tau_1,\ldots,\tau_n)\in (\zo^n)^n$. For each $i\in [n]$, we first test if $D(C(\tau_i))\neq \tau_i$. We can perform this test without modifying $\tau$, as we can evaluate logarithmic depth circuits in $\LOG$. If this holds for some $i$, we return $\tau_i$ without modifying the tape.

    Otherwise, for $i\in [n]$ we invoke the algorithm of \cref{thm:small-width-computable} with the circuit $C$,\footnote{The theorem is stated as taking a logspace-uniform function $f$, we virtually define the function $f_C$ which is trivially defined on inputs of length $m\neq n$ and on inputs of length $n$ applies $C$. Note that the circuit for $f$ is logspace uniform given the input $C$.} and unified tape $\tau_i||0^{O(\log n)}$, where the first $n$ bits correspond to the section of catalytic tape holding $\tau_i$ and we use $O(\log n)$ bits on the worktape for the remainder. The algorithm halts with that section of catalytic tape in configuration $C(\tau_i)||0$. Once we do this for every $i\in [n]$, we shift the tape so that we have $0^n$ at the end and brute force over $y\in \zo^n$ to find some output where $D(C(y))\neq y$, and print the first such $y$. Afterwards, we again use \cref{thm:small-width-computable} to replace $C(\tau_i)||0$ with $D(C(\tau_i))=\tau_i$ on the tape for every $i$, and afterwards halt.
\end{proof}

To prove \cref{thm:small-width-computable}, we will first show that the computation of one layer of a logspace uniform, width-$(n + O(\log n))$, fan-in $2$ circuit can be done in $\inplaceFL$. To do so, we begin by associating any layer of such a circuit with a dependency graph.
\begin{definition}
    Let $C: \bin^n \to \bin^n$ be a logspace uniform, width $w = (n + O(\log n))$, depth $d$, fan-in $2$ circuit. We assume without loss of generality that $C$ has width $w$ at each layer (by padding). We define a dependency graphs $\mathcal{G}_C^\ell$ for $\ell \in [d]$ as follows. $\mathcal{G}_C^\ell$ has vertex set $[w]$ and exactly $w$ edges. An edge exists between vertices $(x, y)$ in $\mathcal{G}_C^\ell$  if and only if some gate in layer $\ell+1$ of $C$ is connected to gates $x$ and $y$ in layer $\ell$ of $C$. Each vertex and edge of $\mathcal{G}_C^\ell$ is in exactly one of two states: a computed state or an uncomputed state.
\end{definition}

$\mathcal{G}_C^\ell$ tells us the dependencies between layer $\ell$ and $\ell+1$ of $C$. It also tells us what information we have gained/lost during our in-place transformation. An edge being in the computed state means that the output of the gate in layer $\ell+1$ corresponding to that edge has been computed and is in memory. Similarly, a vertex being in the computed state means that the output of the gate in layer $\ell$ corresponding to that vertex is currently in memory. We say that we process an edge if we compute it and we process a vertex if we uncompute it. We say that a vertex $v$ is isolated (resp. a leaf) in $\mathcal{G}_C^\ell$ if $v$ is isolated (resp. a leaf) in $\mathcal{G}_C^\ell$ restricted to all vertices and uncomputed edges. An isolated cycle in $\mathcal{G}_C^\ell$ is defined analogously.

$\mathcal{G}_C^\ell$ starts out  with all vertices in the computed state and each edge in the uncomputed state. Our goal is to design an algorithm that implicitly manipulates $\mathcal{G}_C^{\ell}$ in logarithmic space to ultimately arrive at a dependency graph where all vertices are uncomputed and all edges are computed. This corresponds to the state in our algorithm where the tapes (input tape combined with work tape) contain the evaluation of $C$ on the input, up to layer $\ell+1$ of $C$. To achieve this, we rely on the following valid transformations to our dependency graph.

\begin{figure}
    \centering
    \begin{tikzpicture}[scale=1, every node/.style={scale=1}]
  \foreach \i in {1,2,3,4,5,6,7,8,9,10} {
    \node[circle, draw, minimum size=0.8cm] (l\i) at (\i*1.5,0) {$x_{\i}$};
  }
  \foreach \i in {1,2,3,4,5,6,7,8,9, 10} {
    \node[circle, draw, minimum size=0.8cm, fill=gray!20] (u\i) at (\i*1.5,3.5) {$g_{\i}$};
  }
    \draw[->] (l1) -- (u1);
    \draw[->] (l2) -- (u1);
    
    \draw[->] (l2) -- (u2);
    \draw[->] (l4) -- (u2);
    
    \draw[->] (l2) -- (u3);
    \draw[->] (l3) -- (u3);
    
    \draw[->] (l4) -- (u4);
    \draw[->] (l5) -- (u4);
    
    \draw[->] (l1) -- (u5);
    \draw[->] (l5) -- (u5);
    
    \draw[->] (l4) -- (u6);
    \draw[->] (l6) -- (u6);
    
    \draw[->] (l4) -- (u7);
    \draw[->] (l7) -- (u7);
    
    \draw[->] (l1) -- (u8);
    \draw[->] (l7) -- (u8);
    
    \draw[->] (l8) -- (u9);
    \draw[->] (l9) -- (u9);
    
    \draw[->] (l9) -- (u10);
    \draw[->] (l10) -- (u10);

    \draw[->] (l9) -- (u10);
    \draw[->] (l10) -- (u10);

  \node at (0,0) {Layer $\ell$};
  \node at (0,2) {Layer $\ell+1$};
\end{tikzpicture}

\vspace{1cm}

\begin{tikzpicture}[scale=1, every node/.style={circle, draw, minimum size=0.8cm}]
    \node (1) at (0,0) {1};
    \node (2) at (2,1) {2};
    \node (3) at (4,1) {3};
    \node (4) at (2,3) {4};
    \node (5) at (0,3) {5};
    \node (6) at (3,4) {6};
    \node (7) at (1,4) {7};
    \node (8) at (6,0) {8};
    \node (9) at (6,2) {9};
    \node (10) at (6,4) {10};

    \draw (1) -- (2);
    \draw (2) -- (3);
    \draw (2) -- (4);
    \draw (4) -- (5);
    \draw (1) -- (5);
    \draw (4) -- (6);
    \draw (4) -- (7);
    \draw (1) -- (7);
    \draw (8) -- (9);
    \draw (9) -- (10);

\end{tikzpicture}
    \caption{Layer $\ell$ of a circuit $C$ and the corresponding dependency graph $\mathcal{G}_C^{\ell}$}
    \label{fig: circuit_G}
\end{figure}
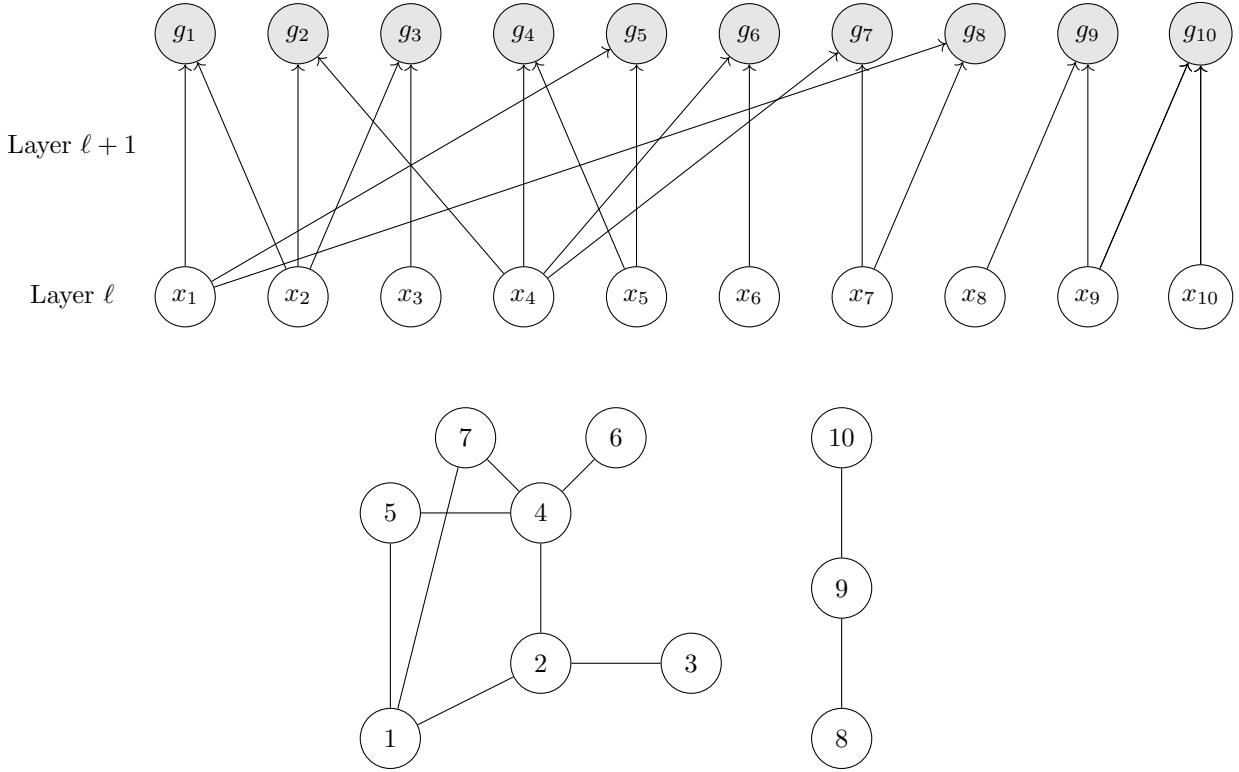

\begin{definition}
    \label{def: valid_transformation}
    Let $\mathcal{G} = (V, E)$ be a dependency graph and $m \in \N$. A valid transformation to $\mathcal{G}$ is an ordered set $\{ o_1, \dots, o_m \}$ where $o_i \in V \cup E$ for all $i \in [1, m]$ such that processing $o_1, \dots, o_m$ (in order) corresponds to one of the following.
    \begin{enumerate}
        \item Uncompute a vertex which has no incident uncomputed edges (an isolated vertex), and then compute an arbitrary edge.
        \item Compute the edge incident to a leaf vertex, and then uncompute that vertex.
        \item For an isolated cycle in $\mathcal{G}$, compute two adjacent edges in the cycle, and then alternate uncomputing vertices and computing edges around that cycle until all vertices are uncomputed and all edges are computed.
\end{enumerate}
\end{definition}

Note that in any series of valid transformations, a vertex is only uncomputed after all its incident edges have been computed. This is a desirable property since it means that we do not uncompute any information from layer $\ell$ which will be necessary in the computation of layer $\ell+1$ gates later. 

Ultimately, these transformations of the dependency graph $\mathcal{G}_C^\ell$ will be performed implicitly by an $\inplaceFL$ algorithm. \cref{lem: nc0_final_perm} shows that being able to implicitly compute a series of valid transformations in logarithmic space will allow us to compute one layer of $C$ in $\inplaceFL$.

\begin{lemma}
    \label{lem: nc0_final_perm}
    Let $f\colon \bin^*\to\bin^*$ be a function with logspace uniform, $\mathsf{NC^0_2}$ circuits $C: \zo^n \rightarrow \zo^n$. Let $\mathcal{G}$ be the dependency graph for $C$. There exists an $\inplaceFL$ algorithm which computes $f$ if there is a series of valid transformations $T_1, \dots, T_t$ which result in all edges of $\mathcal{G}$ being computed and all vertices being uncomputed and $\FL$ algorithms $\pi_1, \pi_2: [n] \rightarrow [n]$ such that the following hold.
    \begin{enumerate}
        \item $\pi_1[i] = j$ if vertex $j$ is the $i^{\text{th}}$ vertex uncomputed in $T_1, \dots, T_t$.
        \item $\pi_2[i] = j$ if edge $j$ is the $i^{\text{th}}$ edge computed in $T_1, \dots, T_t$.
    \end{enumerate}
\end{lemma}
\begin{proof}
Observe that $\pi_1, \pi_2$ satisfy the property that for any $i$, neither input gate feeding into output gate $\pi_{2}[i+2]$ appears among $\pi_{1}[1], \dots, \pi_{1}[i]$. Equivalently, neither of the vertices incident to edge $\pi_{2}[i+2]$ appear in $\pi_{1}[1], \dots, \pi_{1}[i]$. This follows from the fact that our orderings correspond to uncomputing/computing vertices/edges in a valid sequence of i, ii, iii operations in \cref{def: valid_transformation}, and we observe that in any such sequence, at any point, we've always computed at most $2$ more edges than the number of vertices we have uncomputed, so by the time we compute the $(i+2)$th edge we have uncomputed at least $i$ vertices. None of those $i$ vertices can be incident to $\pi_{2}[i+2]$ since we always uncompute the vertices incident to an edge after we have computed that edge. 

We'll use $\pi_{1}$ and $\pi_{2}$ to describe how to get the input tape in-place from the (in-order) sequence of output bits of the gates at level $1$ to the (in-order) sequence of output bits of the gates at level $2$.
We first compute the outputs of gates $\pi_{2}[1]$ and $\pi_{2}[2]$, which we store on the side. Then, for $i \in [1, n-2]$, we'll compute the output of gate $\pi_{2}[i+2]$ and overwrite the $\pi_{1}[i]$th bit of the input tape with the result. By the facts observed above, we know all these values can be computed with logarithmic overhead, and that we never overwrite a bit of the input before it is used to compute a gate's output. So, this process will successfully compute the values of all $\pi_{2}[i]$. They will not end up on the tape in-order --- but, because we can compute $\pi_{1}$ and $\pi_{2}$, we can compute the permutation by which they've been shuffled, so we can then simply perform that permutation in reverse.
\end{proof}

We now show that there exists a series of valid transformations to $\mathcal{G}_C^\ell$ that result in all edges being computed. In fact, we prove the following even stronger claim.
\begin{claim}
    For any dependency graph $\mathcal{G}_C^\ell$, any series of valid transformations to $\mathcal{G}_C^\ell$ result in a dependency graph where all vertices are uncomputed and all edges are computed.
\end{claim}
\begin{proof}
    Let $\hat{\mathcal{G}}$ denote the current state of the dependency graph when restricted to computed vertices and uncomputed edges. Notice first that $\hat{\mathcal{G}}$ initially has average degree 2 since $\hat{\mathcal{G}}$ has exactly the same number of edges and vertices. Furthermore, since in each valid transformation, exactly the same number of vertices are uncomputed as edges are computed, this invariant hold after any series of valid transformations are applied to $\mathcal{G}$. 
    
    Observe that if $\hat{\mathcal{G}}$ is non-empty, then one of the following must hold: either a type i or ii valid transformation can be performed on $\hat{\mathcal{G}}$, or $\hat{\mathcal{G}}$ is a graph with average degree 2 (by the above argument) where every vertex has degree at least 2, which implies that $\hat{\mathcal{G}}$ is the union of disjoint cycles, which implies that a type iii operation can be performed on $\hat{\mathcal{G}}$.

    Therefore, one can always apply a valid transformation to $\mathcal{G}_C^\ell$ which results in $\hat{\mathcal{G}}$ containing fewer vertices and edges. Therefore, any series of valid transformations will eventually result in an empty $\hat{\mathcal{G}}$, implying that all vertices will be uncomputed and all edges will be computed.
\end{proof}

Having shown that such a sequence of operations \emph{exists}, the final step of the argument is to show that we can compute such a sequence with only logarithmic working space. 
That is, we would like to be able to compute in logarithmic space, given an index $i$, the $i$th edge computed in a valid sequence of operations (and, likewise, the $i$th vertex uncomputed).

The first step of our algorithm will be to find a spanning forest of the graph; by results of Reingold and of Nisan and Ta-Shma, this is possible in logarithmic space~\cite{nisan1995symmetric, reingold2008undirected}.
\begin{lemma}[\cite{nisan1995symmetric, reingold2008undirected}]
    \label{lem: msf_comp}
    There exists a logarithmic space algorithm $\mathcal{A}$ which when given oracle access to a graph $\mathcal{G}$ and input $e \in \mathcal{G}$ outputs a value in $\zo$ such that $E = \{ e : \text{$\mathcal{A}^{\mathcal{G}}(e) = 1$}\}$ forms a minimum spanning forest for $\mathcal{G}$. We refer to $E$ as $\text{MSF}(\mathcal{G})$.
\end{lemma}

We now consider the constituent parts of $\mathcal{G}_C^\ell$ and show that determining to which part of $\mathcal{G}_C^\ell$ a particular edge belongs can be done in logarithmic space.
\begin{definition}
    For a connected component of $\mathcal{G}_C^\ell$ we define the following. The \defn{skeleton} to be all edges belonging to $\text{MSF}(\mathcal{G}_C^\ell)$ plus the lexicographically last edge not belonging $\text{MSF}(\mathcal{G}_C^\ell)$. The \defn{feathers} of the connected component are all edges not belonging to the skeleton. The \defn{skull} of the \defn{skeleton} is the cycle in the \defn{skeleton} (if it exists). All edges in the skeleton which are not part of the skull are the \defn{bones}.   
\end{definition}

\begin{definition}
    For each bone in $\text{MSF}(\mathcal{G}_C^\ell)$ , we will define the \defn{prominence} to be the minimum distance in the skeleton from one of its endpoints to a leaf vertex.
\end{definition}

\begin{lemma}
    \label{lem: nc0_subroutines}
    Let $C: \zo^n \rightarrow \zo^n$ be a logspace uniform, $\mathsf{NC^0_2}$ circuit and $\mathcal{G}$ be the dependency graph for $C$. There exist logarithmic space subroutines to determine, for $\mathcal{G}$, which connected component a vertex/edge belongs to, the prominence of an edge and if an edge is a feather, bone, or part of the skull.
\end{lemma}
\begin{proof}
    One can determine which connected component an edge $e$ belongs to simply by identifying each connected component with its lexicographically first edge and then using the logarithmic space algorithm of \cite{reingold2008undirected} to find the lexicographically smallest $v$ such that an endpoint of $e$ is connected to $v$.
    
    The prominence of a bone can be calculated in logarithmic space: the connected component of bones to which a given bone belongs forms a tree; we can determine the minimum depth of a leaf in that tree in logarithmic space by performing a depth-first search.
    
    Determining if an edge is a feather simply requiring confirming that $e \notin \text{MSF}(\mathcal{G}_C^\ell)$ or is the lexicographically last edge not in $\text{MSF}(\mathcal{G}_C^\ell)$. Recall that computing if an edge $e$ is in $\text{MSF}(\mathcal{G})$ can be done in logarithmic space by \cref{lem: msf_comp}. To check if an edge $e$ is in the skull, we start by finding the lexicographically last edge not in $\text{MSF}(\mathcal{G}_C^\ell)$, call it $z$. We then find $e_1 \in \text{MSF}(\mathcal{G}_C^\ell)$ sharing a vertex with $z$, $e_2 \in \text{MSF}(\mathcal{G}_C^\ell)$ sharing a vertex with $e_1$, and so on, until we find that either $e_i = e$, in which case we conclude $e$ is in the skull, or $e_i = z$, in which case we conclude $e$ is a bone.
\end{proof}

With these definitions in place, we are ready to describe how to compute $C$ in $\inplaceFL$.
\begin{proof}[Proof of \cref{thm:small-width-computable}]
Let $w$ be the width of $C$, we assume without loss of generality that $n = w$ and use $O(\log n)$ bits of our work tape to simulate the case whenever we find $w > n$ since $w = n + O(\log n)$. For a fixed $\ell$ consider $\mathcal{G}_C^\ell$. By \cref{lem: nc0_final_perm}, it suffices to show a logarithmic space algorithms which when given $i$, tells us which vertex/edge of $\mathcal{G}_C^\ell$ is the $i^{\text{th}}$ uncomputed/computed via a valid transformations such that after all such transformations have been performed, $\mathcal{G}_C^\ell$ has only computed edges and uncomputed vertices. Our algorithm will then simply apply this subroutine along with \cref{lem: nc0_final_perm} to compute one layer of the circuit $d$ times where $d$ is the depth of the $C$. At any given point, we let $\hat{\mathcal{G}}$ denote the uncomputed edges and computed vertices of the dependency graph.

We will implicitly process the connected components in order of initial average degree. Note that we can find the connected component with the $j$th smallest average degree in logarithmic space. 
That is, our process will compute all edges from one connected component before moving on to the next. Also, whenever we perform a type i operation (i.e. uncomputing an isolated vertex), the edge we will choose to compute will be the lexicographically-first feather in the first connected component with uncomputed feathers.

This ensures that, when we begin processing a component, the uncomputed edges of the component is exactly its skeleton. No edges in the skeleton will have been computed, because the only way edges could have been computed is through type i operations, which only compute feathers. All feathers will have been computed. To see why, consider $\hat{\mathcal{G}}$ at this point and let $c$ denote the average degree of computed vertices in the current connected component at this point. Notice that $c \leq 2$ since if it were not, then $\hat{\mathcal{G}}$ all remaining connected components of $\hat{\mathcal{G}}$ would have average degree greater than 2. But notice that the average degree of $\hat{\mathcal{G}}$ is invariant under valid transformations and started out as $2$. If $c < 2$, then the current connected component is a tree, implying that the current connected component is exactly the skeleton. If $c = 2$, then the current connected component in $\hat{\mathcal{G}}$ is a cycle, which is exactly the skull.

So, we only need to specify how to process the skeleton. The sequence will be: first perform type ii operations to compute all bones, ordered first by prominence and then lexicographically. Then, perform a type iii operation to compute the skull if it exists (starting with the lexicographically first edge in the skull, then its lexicographically first neighbor, and proceeding in that direction around the skull). Finally, perform a type i operation to delete the remaining isolated vertex. Observe that this gives a valid sequence of operations: since we order bones by prominence, we ensure that each bone does indeed touch a leaf at the time of its deletion, and since we handle all bones before the skull we ensure that the skull is an isolated cycle at the time of its deletion.

The crucial point is that now, given an edge of the graph, it's possible to compute in logarithmic space that edge's index in our sequence of edge deletions. 
If the edge is a feather, we can compute its index $j$ among feathers (by counting how many feathers belong to earlier connected components and how many feathers in the same connected component are lexicographically earlier). We therefore know that this feather will be deleted exactly when the $j$th connected component is done being processed, so can find its index in the overall deletion sequence by adding up the total number of vertices in the first $j$ connected components.
If the edge is a bone, we can simply add the number of vertices in earlier connected components plus the number of earlier bones in the same connected components.
If the edge belongs to a skull, we add the number of vertices in earlier connected components, plus number of bones in the same connected component, plus the index of the edge around the skull. 
The same arguments will allow us to determine for any given \emph{vertex} its index in our sequence of \emph{vertex} deletions.
\end{proof}

\section{Computing \ts{\( \FZPP \)}{FP} in-place with an oracle}
\label{sec: fp_in_place}
We show that every function in $\FZPP$ can be computed in-place if our in-place algorithm is given access to a sufficiently powerful oracle.
\TheoremFPInPlace*

We observe that this immediately implies
\cref{cor:sat-in-l-means-inplace-happy} and gives a barrier to proving that $\FL  \not\subseteq \inplaceFL$ and $\FCL \not\subseteq \inplaceFCL$.

\begin{proof}[Proof of \cref{cor:sat-in-l-means-inplace-happy}]
    To show that if $\FCL \not\subseteq \inplaceFL$, then $\NP\neq\LOG$, we will show the contrapositive. If $\NP = \LOG$, then the polynomial hierarchy collapses to $\LOG$, implying $\STP = \LOG$. Then we have $\FCL \subseteq \FZPP \subseteq \inplaceFL^{\STP} = \inplaceFL^{\LOG} = \inplaceFL$ where $\FZPP \subseteq \inplaceFL^{\STP}$ follows from \cref{thm:FPInPlace} and $\STP = \LOG$ follows by assumption.
\end{proof}

\subsection{Good matrices for routing and pseudorandom sets}
\subsubsection{Routing matrices}
At a high level, our strategy will be to transform $x$ into $f(x)$ in-place using a change of basis matrix $Q$. In particular, we will change each coordinate of $Qx$ into $Q f(x)$ one at a time. However, we must be careful to ensure that we can reconstruct $x$ (since it is no longer on the input tape) at each point in the algorithm. To help with that, we introduce a hashing matrix $H$ which will output a short hash of $x$. Still, there may exist two distinct inputs $x, x'$ such that even with $Q$ and $H$ to distinguish between them, our change of basis algorithm becomes confused about whether it is trying to compute $f(x)$ or $f(x')$ at some point in its execution. We call this a conflict.
\begin{restatable}{definition}{fpconflict}
    \label{def: FPConflict}
    Given a function \( f \) and matrices $Q \in \zo^{n \times n}, H \in \zo^{\lceil 2 \log n \rceil \times n}$, we say inputs \( x,x' \in \zo^n \) \emphdef{conflict} if \( f(x') \not= f(x) \) but \( H f(x') = H f(x) \) and for some \( i \in [n] \), $[Q f(x')]_{[1, i]} = [Q f(x)]_{[1, i]}$ and $[Q x']_{[i+1, n]} = [Q x]_{[i+1, n]}$. (Matrix-vector products are modulo 2.)
\end{restatable}

For each function $f$ and input $x$, we want to use $(Q, H)$ such that no conflict occurs. We now formally define conflict avoiding.
\begin{restatable}[conflict-avoiding]{definition}{goodQH}
    \label{def: good_QH}
    Given a function \( f \) and input \( x \in \zo^n \), matrices $Q \in \zo^{n \times n}, H \in \zo^{\lceil 2 \log n \rceil \times n}$ are \emphdef{conflict-avoiding for \( f \) on input \( x \)} if $Q$ is invertible and \( x \) does not conflict with any \( x' \in \zo^n \).
\end{restatable}

Although we cannot guarantee that there exist $(Q, H)$ such that for all $x$, $(Q, H)$ are conflict-avoiding, we can do something almost as good. We can guarantee that there exist a list of $(Q_1, H_1), \dots, (Q_{\poly(n)}, H_{\poly(n)})$ such that for each $x$, one of the $(Q_t, H_t)$ are conflict avoiding for $x$ (we will see how in \cref{subsec: conflict_avoid}). Note that checking which $(Q_t, H_t)$ are conflict avoiding for a given $x$ can be done using only an $\NP$ oracle. This will prove useful later as our oracle algorithm will use an $\NP$ oracle call to decide which $(Q_t, H_t)$ to use as change of basis and hashing matrices by using $\NP$ oracle calls.
\begin{restatable}[universally conflict-avoiding]{definition}{universallyGoodQH}
    \label{def:GoodAdviceForFPInPlace}
    Let $a$ be series of advice strings $a: \N \rightarrow \zo^*$, where we interpret $a(n)$ as a sequence $(Q_1, H_1), \dotsc, (Q_{\poly(n)}, H_{\poly(n)} )$ where $Q_i \in \zo^{n \times n}, H_i \in \zo^{\lceil 2 \log n \rceil \times n}$.
    We say that $a$ is \emphdef{universally conflict-avoiding for \( f \)} if for all sufficiently large \( n \in \N \) and all $x \in \zo^n$, there is a $t \in [1, \poly(n)]$ such that $(Q_t, H_t)$ is conflict-avoiding for \( f \) on input $x$.
\end{restatable}

\subsubsection{Pseudorandom sets for $\FZPP$}
As we are dealing with $\FZPP$ functions which we wish to put into deterministic complexity classes, we will also make use of pseudorandom sets.
\begin{definition}
    We say that a sequence $R = (x_1, \dots, x_m)$ of $n$-bit strings is a pseudorandom set if, for all $n$-input circuits of size $n$:
    \[ \left|\underset{x \sim R}{\mathbb{P}}[C(x) = 1] - \underset{y \sim \zo^n}{\mathbb{P}}[C(y) = 1]\right| \leq 1/n .\]
\end{definition}

Standard probabilistic arguments show that polynomially bounded pseudorandom sets (those where $m = \poly(n)$) exist. The following lemma shows that these can be constructed using an $\AVOID$ oracle call.
\begin{lemma}[\cite{korten2022hardest}]
    \label{lem: PRG_to_Avoid}
    Let PRG be the following search problem: Given $1^n$, output a pseudorandom set $(x_1, \dots, x_m)$ where $x_i \in \zo^n$. Then PRG reduces in polynomial time to $\AVOID$. 
\end{lemma}

We can then specialize this definition to computing pseudorandom set for a specific language on a specific input length.
\begin{definition}\label{def:PRG_for_f}
    Let $c$ and $c'$ be constants. Let $f$ be a $\FZPP$ function and let $A(x; r)$ be the randomized algorithm computing $f$ where $|x| = n$ and $|r| = n^c$ for some constant $c$. Let $C_n: \zo^{\poly(n)} \rightarrow \zo$ be a size $n^{c'}$ circuit which on input $r' \in \zo^{n^{c'}}$ outputs whether $A(x; r'_{[1, n^c]}) \neq \bot$ for $x\in \zo^n$. We say that a sequence $R = (x_1, \dots, x_m)$ is a pseudorandom set for $f$ (on length $n$) if the following holds:
    \[ \left|\underset{x \sim R}{\mathbb{P}}[C_n(x) = 1] - \underset{y \sim \zo^{n^{c'}}}{\mathbb{P}}[C_n(y) = 1]\right| \leq 1/n^{c'} \]
\end{definition}

\begin{lemma}
    \label{lem: PRG_for_f}
    Let $\text{PRG}^f$ be the following search problem: Given $1^n$, output a pseudorandom set for $f$. $\text{PRG}^f$ reduces in polynomial time to $\AVOID$. 
\end{lemma}
\begin{proof}
    Follows directly from \cref{lem: PRG_to_Avoid}.
\end{proof}

\subsection{Computing \ts{\( f \)}{f} in-place using advice}\label{sec:RandomBasisAlgorithmCore}
For our result, we will use the slightly non-standard notion of Turing machine with advice where we parameterize by the advice. Note that the following definition $\NP/a$ is clearly a subset of $\NP/\poly$ since we are simply restricting ourselves to considering \emph{one} sequence of advice strings rather than all polynomially bounded advice.
\begin{definition}\label{def:NPWithParticularAdvice}
    Let $a: \N \rightarrow \zo^*$ be any function such that $|a(n)| = \poly(n)$. We define the class $\NP/a$ as follows. A language $\mathcal{L}$ is in $\NP/a$ if and only there exists a non-deterministic, polynomial time Turing machine $M$ such that $x \in \zo^n$ is in $\mathcal{L}$ if and only if $M(x, a(n)) = 1$.
\end{definition}

\begin{lemma}
    \label{lem: inplace-f-in-NP/a}
    Let $f$ be a length-preserving function in \( \FZPP \) and $a: \N \rightarrow \zo^*$ be universally conflict-avoiding for $f$ (\cref{def:GoodAdviceForFPInPlace}) and $R : \N \rightarrow \zo^*$ be a sequence of pseudorandom sets for $f$ (\cref{def:PRG_for_f}).
    Then $f \in \inplaceFL^{\NP/(a, R)}$.
\end{lemma}
\begin{proof}
We describe an algorithm for computing \( f \) in-place, given access to an \( \NP/(a,R) \) oracle.
Let \( x \in \zo^n \) be any input.
Let $A(x; r)$ be randomized algorithm for $f$ where $|x| = n$ and $|r| = n^c$. Let $R(n) = (r_1, \dots, r_m)$. We first determine an $i \in [m]$ such that
$A(x; {r_i}_{[1, n^c]}) \neq \bot$. From now on, we will view $f$ as being computable in polynomial time with the knowledge that we can do so by computing $A(x; {r_i}_{[1, n^c]})$.

Recall that we interpret the advice string \( a(n) \) as encoding a sequence of \( \poly(n) \) pairs \( ( Q_t, H_t ) \), where we will attempt to use \( Q_t \in \zo^{ n \times n } \) as a change of basis and \( H_t \in \zo^{\lceil 2 \log n \rceil \times n} \) as a hash function, and that one of these pairs is guaranteed to be \emph{conflict-avoiding} for \( f \) on input \( x \).

Our algorithm begins by finding such a conflict-avoiding \( (Q_t,H_t) \); it can do this using the \( \NP/a \) oracle because the problem of determining whether \( (Q_t,H_t) \) is conflict-avoiding is in \( \CoNP \).
Henceforth, let \( (Q,H) = (Q_t,H_t) \) be the conflict-avoiding pair the algorithm found.

Next, the algorithm stores \( H f(x) \in \zo^{\lceil 2 \log n \rceil} \) in memory.
This serves as a hash to help make sure we always have enough information to reconstruct \( f(x) \).

Now, for \( i = 0, \dotsc, n \), define \( v_i(x) \in \zo^n \) to be the unique vector such that the first \( i \) coordinates of $Q v_i(x)$ match the first \( i \) coordinates of \( Q f(x) \), and the remaining coordinates match the corresponding coordinates of \( Q x \).
Then \( v_0(x) = x \) and \( v_d(x) = f(x) \), and the algorithm proceeds in \( n \) steps, transforming \( v_i(x) \) into \( v_{i+1}(x) \) for \( i = 0, \dotsc, n-1 \).

At each step, the fact that \( (Q,H) \) is conflict-avoiding for \( f \) on input \( x \) means the algorithm has enough information to determine \( v_{i+1}(x) \) given \( v_i(x) \) and the hash \( H f(x) \).
What remains is to show how to change \( v_i(x) \) into \( v_{i+1}(x) \) \emph{in place}.
Note that their difference is \( ( z_i - y_i ) (Q^{-1})_i \), where \( z_i \) and \( y_i \) are the \( i \)-th coordinates of \( Q f(x) \) and \(  x \) respectively and \( (Q^{-1})_i \) is the \( i \)-th column of \( Q^{-1} \).
So, we make the change as follows.
First, compute \( z_i \) and \( y_i \) (the next paragraph shows how) and store their difference on the work tape.
Then add that value times \( (Q^{-1})_i \) to the value stored on the input/output tape, using the \( \NP \) oracle to compute matrix inverse.

Finally we show how to compute \( z_i \) and \( y_i \).
The latter is just the \( i \)-th coordinate of \(  v_i(x) \), which can be computed in \( \FL \).
We can compute \( z_i \) with the help of the \( \NP \) oracle: to compute the \( k \)-th bit of \( z_i \), we ask the oracle the following question: does there exist some \( x \) which is consistent with the available information (\( v_i(x) \) and the stored hash \( H f(x) \)) and such that the \( k \)-th bit of \( v_{i+1}(x) \) is \( 1 \)?
The answer is ``yes'' iff the \( k \)-th bit of \( z_i \) is one, because we have ensured there is exactly one \( v_{i+1}(x) \) which is consistent with the available information.
\end{proof}

\subsection{Conflict-avoidance from an \ts{$\AVOID$}{AVOID} oracle}
\label{subsec: conflict_avoid}

Having shown that given access to advice $a$ which is universally conflict-avoiding for $f$ and a pseudorandom set $R$ we can compute $f$ in $\inplaceFL^{\NP/(a,R)}$, we now consider the complexity of computing the advice $a$ for ourselves. To that end, we define the following problem and show that it can be solved given access to an $\AVOID$ oracle.
\begin{definition}
    Let $f$ be a length preserving function. We define the search problem $\ROUTE^f$ as follows. Compute any function $a': \{ 1 \}^* \rightarrow \zo^*$ such that $a(n): \N \rightarrow \zo^*$, $a(n) = a'(1^n)$ is universally conflict-avoiding for $f$.
\end{definition}
\begin{lemma}
    \label{lem: advice_to_avoid}
    Let $f$ be a length preserving function. $\ROUTE^f$ can be solved in polynomial time given access to an $\AVOID$ oracle and a family of polynomial-size circuits computing $f$.\footnote{We obtain this family of circuits from hardwiring the randomness $R$, which we produce separately.}
\end{lemma}

To prove this lemma, we will describe a family of expanding circuits \( (C_n)_{ n \in \N } \) whose output we interpret as a sequence of pairs \( (Q_t, H_t) \).
We design the circuits so that their range includes every sequence which is not universally collision-avoiding for \( f \), and use the \( \AVOID \) oracle to solve our problem.
Before going into the details, we explain some intuition behind why this should be possible.

A first requirement for a sequence \( (Q_t, H_t) \) to be universally conflict-avoiding is that each matrix \( Q_t \) is invertible.
We handle this using by having the circuit \( C_n \) output each \( Q_t \) using a compact encoding given by \cref{lem:EncodeInvertible} which admits only invertible matrices.

The more difficult requirement is that for every possible input \( x \in \zo^n \), some pair \( (Q_t, H_t) \) ensures that there is no conflicting \( x' \in \zo^n \).
For \( x \) and \( x' \) to conflict means that \( H_t f(x) = H_t f(x') \) and for some \( i \), the first \( i \) coordinates of \( Q_t f(x) \) and \( Q_t f(x') \) match, and the last \( n-i \) coordinates of \( Q_t x \) and \( Q_t x' \) also match.
To complete the proof of \cref{lem: advice_to_avoid}, we show that these linear constraints imply there is a way to efficiently encode any such \( (Q_t, H_t) \) as inputs to the circuit \( C_n \), forcing the \( \AVOID \) oracle to output at least one conflict-avoiding pair \( (Q_t, H_t) \).

Before providing the proof in detail, we introduce two useful lemmas.
The first one will help us reduce the problem of finding a universally conflict-avoiding sequence to finding a pair \( (Q,H) \) which is conflict-avoiding on some given input \( x \). We note that essentially the same reduction appeared in \cite{chen2023range} (as a reduction from the problem of generating partially hard truth tables to $\AVOID$).

\begin{lemma}
    \label{lem: list_avoid}
    Let $C_x: \zo^{n} \rightarrow \zo^{n+1}$ be a family of efficiently computable (and index-able) circuits indexed by $x \in \zo^{\poly(n)}$. There exists an $\FP^{\AVOID}$ algorithm which takes as input $1^n$ and outputs $(y_1, \dots, y_{\poly(n)})$ such that for all $x$, there exists an $i$ such that $y_i \notin \text{Range}(C_x)$.
    
\end{lemma}
\begin{proof}
    Say $x \in \zo^{f(n)}$. Define a new circuit $C': \zo^{f(n)} \times (\zo^n)^{f(n)+1} \rightarrow (\zo^{n+1})^{f(n)+1}$ by $C'(x, z_1, \dots, z_{f(n)+1}) = (C_x(z_1), \dots, C_x(z_{f(n)+1}))$. The reduction to $\AVOID$ calls $\AVOID$ on $C'$ to obtain $(y_1, \dots, y_{f(n)+1})$ which it then outputs.

    The reduction clearly proceeds in polynomial time. $C'$ is expanding since $(n+1)(f(n)+1) > f(n) + n (f(n)+1) $. To see correctness, say for the sake of contradiction that there exists an $x$ such that for all $y_i$, $y_i \in \text{Range}(C_x)$. For each \( i \), choose some $z_i \in C_x^{-1}(y_i)$. Clearly, $C(x, z_1, \dots, z_{f(n)+1}) = C_x(z_1, \dots, z_{f(n)+1}) = (y_1, \dots, y_{f(n)+1})$. This contradicts the fact that $(y_1, \dots, y_{f(n)+1})$ is avoided by $C'$.
\end{proof}

One requirement of a conflict-avoiding pair \( (Q,H) \) is for \( Q \) to be invertible.
The following lemma helps us meet this requirement.
\begin{lemma}\label{lem:EncodeInvertible}
    There is a bijection between the set of invertible (over \( \GF(2) \)) matrices in \( \zo^{ n \times n } \) and the set of bounded integer sequences in \( R_n := [2^n - 2^0] \times [2^n - 2^1] \times \dotsb \times [2^n - 2^{ n-1 }] \), for every \( n \in \N \).
    Moreover, this family of bijections and their inverses are in \( \FP \).
\end{lemma}
\begin{proof}
Given \( (a_1, \dotsc, a_n) \in R_n \), we construct the corresponding invertible \( A \in \zo^{ n \times n } \) one column at a time.
Suppose we have already constructed the first \( i \) columns \( A_1, \dotsc, A_i \).
Then we must choose \( A_{ i+1 } \) from the set \( S_{ i+1 } \zo^n \setminus \Span \{ A_1, \dotsc, A_i \} \) based on \( a_{ i+1 } \in [2^n - 2^i] \).

To do this, first temporarily append \( n-i \) columns \( A'_{i+1}, A'_{i+2}, \dotsc, A'_n \) to \( A \) to produce an invertible matrix \( Q = A_1, \dotsc, A_i, A'_{i+1}, \dotsc, A'_n \).
Any deterministic way of doing this will work; for example, we could at each point take the next column to be the first standard basis vector which is not yet spanned.
Now, interpret \( a_{ i+1 } + 2^{ i+1 } \) as a base-2 number \( b = (b_1, \dotsc, b_n) \in \zo^n \), from least to most significant digit.
Because we added \( 2^{ i+1 } \), the last \( n-i \) bits \( b_{ i+1 }, \dotsc, b_n \) can't all be zero.
Therefore, the vector \( Q b \) is not in the span of \( A_1, \dotsc, A_i \): that is, \( Q b \in S_{ i+1 } \). 
So, we may take \( A_{ i+1 } = Q b \).
\end{proof}

Now we are ready to prove \cref{lem: advice_to_avoid}.

\begin{proof}[Proof of \cref{lem: advice_to_avoid}]
    Fix $x \in \zo^n$. We will first construct an \( \AVOID \) instance $C$ which outputs $Q \in \zo^{n \times n}, H \in \zo^{\lceil 2 \log n\rceil \times n}$ which are conflict-avoiding for $f$ on input $x$.
Define $v_i(Q, x) \in \zo^n$ so that its first $i$ elements are $Q f(x)$ and its last $n-i$ elements are $Q x$.
Recall that we say that $Q \in \zo^{n \times n}, H \in \zo^{\lceil 2 \log n\rceil \times n}$ are conflict-avoiding for $f$ on input $x$ if $Q$ is invertible and no \( x' \in \zo^n \) conflicts with \( x \): that is, for all \( x' \in \zo^n \), either $f(x) = f(x')$, or $H f(x) \neq H f(x')$, or $v_i(Q, x) \neq v_i(Q, x')$.
    
Let \( R = [2^n - 2^0] \times [2^n - 2^1] \times \dotsb \times [2^n - 2^{ n-1 }] \).
Using \cref{lem:EncodeInvertible}, we may freely switch between elements of \( R \) and their corresponding matrices with the confidence that such maps can be implemented efficiently. We are now ready to define the circuit $C$ which we will feed to our $\AVOID$ oracle.
    \[ C: \zo^n \times [n] \times \zo^{{\lceil 2 \log n\rceil \times (n-1)}} \times \zo^{n \times (n-1)} \rightarrow R \times \zo^{{\lceil 2 \log n\rceil \times n}} \ldotp \]

    $C$ interprets its input as some $x' \in \zo^n$, $i \in [n]$, $H^* \in \zo^{\lceil 2 \log n\rceil \times (n-1)}$, $b_1, \dots, b_n \in \zo^{n-1}$.
   One should think of $H^*$ as an incomplete matrix missing one column which our circuit will fill in, and $b_1, \dots, b_n$ as indexing rows of the matrix $Q$.

    If $f(x) = f(x')$, $C$ outputs $0$.
    Otherwise, $f(x) \not= f(x')$, and \( C \) will construct a pair \( (Q,H) \) with respect to which \( x \) and \( x' \) conflict (\cref{def: FPConflict}), meaning \( v_i( Q, x) = v_i( Q, x' ) \) and \( H f(x) = H f(x') \).
    We will design \( C \) carefully so that \emph{every} pair \( (Q,H) \) which is not conflict-avoiding for \( f \) on input \( x \) will arise in this way.

    On input $(x', j, H^*, (b_1, \dots, b_n))$, \( C \) proceeds as follows.
    Let \( m \) be the index of the first non-zero bit of \( f(x) \oplus f(x') \) (where \( \oplus \) denotes coordinate-wise addition modulo \( 2 \)).

    To construct \( H \), take the \( n-1 \) columns of the input \( H^* \), and insert as the \( m \)-th column a specially-constructed vector \( r \in \zo^{\lceil 2 \log n \rceil} \) so that \( H f(x) = H f(x') \).
    There is exactly one choice of \( r \) which satisfies this: \( C \) takes \( r = H^* \cdot z_{[n] \setminus \{ m \}} \) where \( z = f(x) \oplus f(x') \).
    
    Next, \( C \) determines the matrix \( Q \) row by row based on the input vectors \( b_1, \dotsc, b_n \).
    For nonzero \( \alpha \in \zo^n \), let $\phi_\alpha: \zo^{n-1} \rightarrow \zo^n$ be any efficiently computable bijective map whose image is all vectors orthogonal to $\alpha$.
    (For example, determine the index of the first nonzero coordinate of \( \alpha \) and have \( \phi_\alpha \) use its input as the other \( n-1 \) coordinates.)
    If $i \leq j$, \( C \) takes $Q_i = \phi_{f(x) \oplus f(x')}(b_i)$.
    If $i > j$, \( C \) takes $Q_i = \phi_{x \oplus x'}(b_i)$.
    If $Q$ is not invertible, $C$ outputs $0$. Otherwise, $C$ outputs $(H, Q)$.

    Now, we show that the range of \( C \) includes all pairs \( (Q,H) \) which are not conflict-avoiding (but where \( Q \) is invertible).
    Indeed, suppose not: then some \( x' \in \zo^n \) conflicts with \( x \) with respect to \( (Q,H) \).
    Let \( m \) be the index of the first nonzero bit in \( f(x) \oplus f(x') \), and set the input \( H^* \) to be \( H \) with its \( m \)-th column deleted.
    Then \( C \) will output \( H \) because it is the only way to extend \( H^* \) that satisfies \( H f(x) = H f(x') \).
    Similarly, let the input \( i \) be an index such that \( v_i(Q,x) = v_i(Q,x') \).
    Then the first \( i \) rows of \( Q \) are orthogonal to \( f(x) \oplus f(x') \), and so they are each in the range of the map \( \phi_{ f(x) \oplus f(x') } \), and similarly, the remaining rows are orthogonal to \( x \oplus x' \), and therefore in the range of \( \phi_{ x \oplus x' } \).
    So, we may set \( b_1, \dotsc, b_n \) to be the pre-images of the rows under \( \phi_{ f(x) \oplus f(x') } \) or \( \phi_{ x \oplus x' } \) as appropriate, and \( C \) will output \( Q \).

    We now show that $C$ is expanding by at least 1 bit when its input and output are interpreted as bits. As noted in \cref{subsec: compression_classes}, it suffices to show that, when interpreted as integer, the size of the codomain of $C$ is at least 8 times the size of the domain of $C$.
    \begin{align*}
        \frac{|R| \cdot 2^{\lceil 2 \log n \rceil \times n}}{2^n \cdot n \cdot 2^{\lceil 2 \log n \rceil \times (n-1)} \cdot 2^{n \times (n-1)}}
        &= \frac{|R| \cdot 2^{\lceil 2 \log n \rceil}}{2^{n^2} \cdot n}\\
        &= \frac{\prod_{j=0}^{n-1} (2^n - 2^j) \cdot 2^{\lceil 2 \log n \rceil}}{2^{n^2} \cdot n}\\
        &\geq \frac{n \cdot \prod_{j=0}^{n-1} (2^n - 2^j)}{2^{n^2}}\\
        &= n \cdot \prod_{k=1}^{n} \left( 1 - \frac{1}{2^k} \right)\\
        &= \Omega(n)
    \end{align*}
    The second equality follows from the size of $R$ established in \cref{lem:EncodeInvertible}. Since $\Omega(n)$ is clearly larger than 8 for sufficiently large $n$, the circuit $C$ is sufficiently expanding.

    Having constructed a circuit which has in its range all \( (Q,H) \) which are not conflict-avoiding, we are ready to complete the proof by solving the \( \ROUTE^f \) problem.
    Note that our circuit \( C \) depends on the input \( x \).
    We invoke \cref{lem: list_avoid} to compute a sequence of \( \Poly(n) \) pairs \( (Q_t,H_t) \) such that one is conflict-avoiding for every \( x \): that is, the sequence is universally conflict-avoiding for \( f \), as required.
\end{proof}

\subsection{Putting it all together}
We are now ready to prove that we can compute $f$ in-place given access to sufficiently powerful oracles.
\begin{proof}[Proof of \cref{thm:FPInPlace}]
    We first show that $f$ can be computed in $\inplaceFL^{\NP/\poly}$. By \cref{lem: inplace-f-in-NP/a}, we know $f$ can be computed in $\inplaceFL^{\NP/(a, R)}$ for any function $a: \mathbb{N} \rightarrow \zo^n$ which outputs universally conflict-avoiding values for $f$ and sequence of pseudorandom sets for $f$ $R:\N\ra\zo^*$. We also know that such an $a$ exists by \cref{lem: advice_to_avoid} and that such an $R$ exists by \cref{lem: PRG_for_f}. Therefore, $f$ can be computed in $\inplaceFL^{\NP/\poly}$.

    To show that $f \in \inplaceFL^{ \STP }$, we begin by observing that by \cref{lem: inplace-f-in-NP/a}, \cref{lem: advice_to_avoid}, and \cref{lem: PRG_for_f}, $f$ can be computed in $\inplaceFL$ given oracle access to some language $\mathcal{L}$ which can be decided by a polynomial-time algorithm using oracle access to a unique valued $\AVOID$ (one which outputs the same answer for any input) oracle and $\NP$ oracle. We show that any such language $\mathcal{L}$ is in $\STP$. Since $\STP$ is Turing closed, it suffices to show that $\mathcal{L}$ can be decided using oracle access to $\STP$. On input $x \in \zo^n$, we decide if $x \in \mathcal{L}$ as follows. \cref{cor: svAvoid_using_S2P} tells us that there exists a $\FP^{\STP}$ algorithm that solves $\AVOID$. By \cref{lem: PRG_for_f}, we can use this algorithm to obtain a pseudorandom set for $f$, and hence build a polynomial-size circuit $C_n'$ where $C_n'(x)=f(x)$ (this circuit enumerates over the elements of the pseudorandom set $R$ and simply takes the first non-$\perp$ output).
    
    Next, by \cref{lem: advice_to_avoid}, we can use this algorithm to solve $\ROUTE^f$ (where we provide the circuit $C_n'$) and get advice $q \in \zo^{\poly(n)}$. Finally, we use our final $\STP$ oracle call to simulate the $\NP$ call (which uses the advice $q$). This whole computation is in $\P^{\STP} = \STP$~\cite{russell1998symmetric}, as desired.
\end{proof}

\section{Computing linear transformations in-place}\label{sec:in_place_matrix}

Here we describe how to do matrix-vector multiplication in-place: given read-only access to a matrix \( A \) over a representable field (\cref{def:RepresentableField}), we replace a vector \( v \) stored in memory with \( Av \).
The algorithm uses \( O( \log n ) \) free space and uses an oracle that can compute certain functions which are in \( \CL \); this puts matrix-vector multiplication in \( \inplaceFL^{ \CL } \) and therefore in \( \inplaceFCL \).

\matvecprod*

We describe our algorithm in the following sections.

In \cref{sec:InPlaceMatrixAlmostTriangular}, we handle a special case where \( A \) is what we call an \emph{almost upper-triangular} matrix.
In \cref{sec:InPlaceMatrixBasisChange}, we relax this requirement to \( Q^{-1} A Q \) being almost upper-triangular for some invertible \( Q \).
In \cref{sec:InPlaceMatrixQ} we describe such a matrix \( Q \), and in \cref{sec:InPlaceMatrixComputeQ} we show it can be computed in \( \FCL \), so that our algorithm can make use of it.

\subsection{Multiplying by almost-triangular matrices}\label{sec:InPlaceMatrixAlmostTriangular}

If \( U \in K^{ n \times n } \) is upper-triangular, we can transform \( x \) into \( Ux \) one coordinate at a time, using the fact that each coordinate \( (Ux)_i \) only depends on coordinates \( v_j \) for \( j \ge i \).
For example, if
\[ U = \begin{pmatrix} 2 & 1 & 5 \\ 0 & 4 & 2 \\ 0 & 0 & 3 \end{pmatrix} \]
our in-place computation might follow these steps:
\[
	\newcommand{\Hl}[1]{\textcolor{red}{#1}}
	x = \begin{pmatrix} 1 \\ 4 \\ 5 \end{pmatrix}
	\rightarrow
	\begin{pmatrix} \Hl{31} \\ 4 \\ 5 \end{pmatrix}
	\rightarrow
	\begin{pmatrix} \Hl{31} \\ \Hl{26} \\ 5 \end{pmatrix}
	\rightarrow
	\begin{pmatrix} \Hl{31} \\ \Hl{26} \\ \Hl{15} \end{pmatrix}
	=
	Ux
\]
At each step, the newly computed coordinate does not depend on any of the coordinates that have already been changed.
(For this example, assume the field is \( K = \GF(p) \) for some prime \( p > 31 \).)

We extend this to \emph{almost upper-triangular} matrices:
\begin{definition}
A matrix \( U \) is \emphdef{almost upper-triangular} if every nonzero entry \( U_{i,j} \) has \( i \le j + 1 \).
\end{definition}
The following algorithm replaces \( x \) with \( Ux \) when \( U \) is almost upper-triangular.
The variables \( v_1, \dotsc, v_n \) are the values currently held in memory; at the start, \( v = x \) and at the end \( v = Ux \).

To see that this algorithm is correct, observe that at the start of every iteration of the outer for loop:
\begin{itemize}
\item the first \( i-1 \) coordinates of \( v \) match the first \( i-1 \) coordinates of \( Ux \);
\item the remaining coordinates match the corresponding coordinates from \( x \); and
\item if \( i > 1 \), the variable ``previous'' holds \( x_{ i-1 } \).
\end{itemize}

\begin{algorithm}[H]
\( \mathrm{carry} \gets 0 \)\;
\For{\(i=1..n\)}{
	\( \mathrm{current} \gets v_i \)\tcc*[l]{\( \mathrm{current} = x_i \)}
	\( v_i \gets 0 \)\;
	\tcc{Handle entries at or above the diagonal: \( U_{i,j} \) where \( i \le j \).}
	\For{\(j=i..n\)}{
		\( v_i \gets v_i + U_{i,j} v_j \)\;
	}
	\tcc{Handle the below-diagonal entry \( U_{i-1,i} \).}
	\If{\( i>1 \)}{
		\( v_i \gets v_i + U_{ i-1, i } \cdot \mathrm{previous} \)\;
	}
	\( \mathrm{previous} \gets \mathrm{current} \)\;
}
\caption{\label{alg:AlmostUpperInPlace}Multiplying an almost-triangular matrix \( A \) in-place.}
\end{algorithm}

\subsection{Changing bases}\label{sec:InPlaceMatrixBasisChange}

Now, suppose have access to an invertible matrix \( Q \) such that \( Q^{ -1 } A Q \) is almost upper-triangular.
Our goal is to somehow use the fact that \cref{alg:AlmostUpperInPlace} can correctly multiply by the matrix \( U = Q^{ -1 } A Q \) in-place to instead multiply by \( A \) in-place.

Let \( w \in K^n \) be the values currently held on the catalytic tape: at the start, \( w=x \), and at the end, our goal is to guarantee \( w=Ax \).
We achieve this by running a simulation of \cref{alg:AlmostUpperInPlace}, with all of its accesses to the matrix \( U \) and vector \( v \) transformed into accesses to \( A \) and \( w \) via a change of basis.

The vector \( v \) presented to \cref{alg:AlmostUpperInPlace} will always correspond to \( Q^{ -1 } w \), and the matrix \( U \) will be \( Q^{ -1 } A Q \).
At the beginning, \( w = x \), and so the vector \( v \) that the algorithm sees is \( Q^{ -1 } w = Q^{ -1 } x \).
Then, since \cref{alg:AlmostUpperInPlace} is correct, it transforms \( v \) from \( Q^{ -1 } x \) to \( U Q^{ -1 } x = Q^{ -1 } A Q Q^{ -1 } x = Q^{ -1 } A x \).
Since the final value of \( v = Q^{ -1 } w \) is \( Q^{ -1 } A x \), it follows that the final value of \( w \) is \( A x \).

To do this, the simulation operates as follows:
\begin{itemize}
\item When \cref{alg:AlmostUpperInPlace} wishes to read an entry \( U_{ i,j } \), we supply it with the \( (i,j) \)-th entry of \( Q^{ -1 } A Q \).
\item When it wishes to read \( v_i \), we supply it with the \( i \)-th coordinate of \( Q^{ -1 } w \).
\item When it wishes to write a value \( a \in K \) to \( v_i \), we instead modify \( w \) in such a way that the \( i \)-th coordinate of \( Q^{ -1 } w \) becomes \( b \) and all other coordinates are unchanged. To do this: let \( b \) be the previous value of the \( i \)-th coordinate, and then add \( a-b \) times the \( i \)-th column of \( Q \) to \( w \). This has the effect of adding \( (a-b) Q^{ -1 } Q e_i = (a-b) e_i \) to \( Q^{ -1 } w \), where \( e_i \) is the \( i \)-th standard basis vector.
\end{itemize}
All of the above operations can be done by a logspace machine with a \( \CL \) oracle, since by \cref{lem:LinAlg} matrix inverses can be computed in \( \FCL \).

\subsection{A particular change of basis}\label{sec:InPlaceMatrixQ}

Next we describe a particular basis change \( Q \) such that \( Q^{ -1 } A Q \) is almost upper-triangular.
Although we eventually want to show \( Q \) can be computed in \( \FCL \), for now we will find it simpler to describe an algorithm that generates the columns of \( Q \) one at a time, each based on the previous columns.

\subsubsection{Warm-up special case}

The problem becomes simpler if, for some ``starting'' vector, say the standard basis vector \( e_1 \), the vectors \( e_1, A e_1, A^2 e_1, \dotsc, A^{ n-1 } e_1 \) are linearly independent.
In this case, we use these vectors as the \( n \) columns of \( Q \):
\[
	Q =
	\begin{pmatrix}
	e_1 & A e_1 & A^2 e_1 & \cdots & A^{ n-1 } e_1
	\end{pmatrix}
\]
Then for \( i<n \), the \( i \)-th column of \( Q^{ -1 } A Q \) is \( e_{ i+1 } \), the \( (i+1) \)-st standard basis vector.
To see this: note that \( Q^{ -1 } A Q e_i = Q^{ -1 } A ( A^{ i-1 } e_1 ) = Q^{ -1 } A^i e_1 \).
Observe that \( A^i e_1 \) is the \( (i+1) \)-st column of \( Q \), so multiplying by \( Q^{ -1 } \), we get \( e_{ i+1 } \).
So,
\[
	Q^{ -1 } A Q
	=
	\begin{pmatrix}
	e_2 & e_3 & e_4 & \cdots & e_n & A e_n
	\end{pmatrix}
\]
which is almost upper-triangular.

\subsubsection{In general}

Unfortunately, it is possible that the matrix \( Q \) from the warm-up is not invertible.
This happens if \( A^t b_0 \) is a linear combination of the previous vectors for some \( t < n \): \( A^t e_1 = \sum_{ i=1 }^{ t-1 } c_i A^i e_1 \) for some coefficients \( (c_i) \).

If this happens, we will keep the first \( t \) columns, and then start the process again, with a new starting vector.
For example, if \( e_2 \) is not spanned by \( \{ A^i e_1 \}_{ i=0 }^{ t-1 } \), then we will try this as our matrix \( Q \):
\[
	Q =
	\begin{pmatrix}
	e_1 & A e_1 & \cdots & A^{ t-1 } e_1 & e_2 & A e_2 & \dotsc & A^{ n-t-1 } e_2
	\end{pmatrix}
\]
If these columns are all linearly independent, then we will have succeeded: \( Q^{ -1 } A Q \) is almost upper-triangular, by an argument similar to the one in the warm-up.
But there could be a linear dependence here too.

We solve this generally by switching to a new standard basis vector every time a linear dependence would have been introduced.
Here is the algorithm:

\begin{algorithm}[H]
Initialize \( L \) to be an empty list.\;
\For{\(i=1..n\)}{
	\( j \gets 0 \)\;
	\While{\( A^j e_i \) is not spanned by vectors in \( L \)}{
		Append \( A^j e_i \) to \( L \).\label{line:AppendAVectorToQ}\;
		\( j \gets j+1 \)\;
	}
}
Set \( Q \) to be the matrix whose columns are the vectors in \( L \), in the same order.\;
\caption{\label{alg:ComputeQSequentially}Computing the matrix \( Q \) sequentially.}
\end{algorithm}

It remains to show that \( Q^{ -1 } A Q \) is almost upper-triangular.
For this, it suffices to show that for any \( k \in \{ 1, \dotsc, n \} \), the last \( n-k-1 \) coordinates of \( Q^{ -1 } A Q e_k \) are zero.
Now, \( Q e_k \) equals \( A^j e_i \) where \( i,j \) have the values they did the \( k \)-th time \cref{line:AppendAVectorToQ} was executed.
So, \( Q^{ -1 } A Q e_k = Q^{ -1 } A^{ j+1 } e_i \).
There are two cases to consider: either \( A^{ j+1 } e_i \) is spanned by the first \( k \) columns of \( Q \), or it isn't.
If it is, then this vector only has nonzero entries in the first \( k \) coordinates.
(These values correspond to the way \( A^{ j+1 } e_i \) appears as a linear combination of the previous columns.)
If it is not spanned, then the next execution \cref{line:AppendAVectorToQ} appends \( A^{ j+1 } e_i \) to \( L \), so the \( (k+1) \)-st column of \( Q \) is \( A^{ j+1 } e_i \), and so \( Q^{ -1 } A Q e_k = e_{ k+1 } \).

\subsection{Computing a good basis in \ts{\( \FCL \)}{FCL}}\label{sec:InPlaceMatrixComputeQ}

It remains to be shown that any entry of the matrix computed by
\cref{alg:ComputeQSequentially} can be computed in \( \FCL \).

For any \( i \), let \( r_i \) be the rank of the \( n \times n \) matrix whose columns are \( A^j e_{i'} \) for \( i' \) ranging from \( 1 \) to \( i \) and \( j \) ranging from \( 0 \) to \( n-1 \).
Then \( r_i \) equals the dimension of the vector space spanned by \( L \) after the \( i \)-th iteration of the for loop in \cref{alg:ComputeQSequentially}.
(Define \( r_0 = 0 \).)
\( r_i \) can be computed in \( \FCL \) by \cref{lem:LinAlg}.

For any \( k \), the \( k \)-th column of \( Q \) can be found as follows.
When \cref{line:AppendAVectorToQ} of \cref{alg:ComputeQSequentially} found the \( k \)-th column of \( Q \), its value of \( i \) was the largest \( i \) such that \( r_{ i-1 } < k \), and \( j \) was \( k-r_{ i-1 } \). 
Compute \( i \) in \( \FCL \) by computing each \( r_i \) in turn.
Then the \( k \)-th column of \( Q \) is \( A^{ k-r_{ i-1 } } e_i \).

Finally, we can prove \cref{thm:mat-vec-prod}.
\begin{proof}[Proof of \cref{thm:mat-vec-prod}]
    Given a matrix \( A \) as input, our algorithm's goal is to multiply a vector by \( A \) in-place.

    Let \( Q \) be the matrix described in \cref{sec:InPlaceMatrixQ}, so that \( Q^{ -1 } A Q \) is almost upper-triangular.
    As explained in \cref{sec:InPlaceMatrixComputeQ}, \( Q \) can be computed in \( \FCL \).
    
    Since \( Q^{ -1 } A Q \) is almost upper-triangular, \cref{alg:AlmostUpperInPlace} can multiply a vector by \( Q^{ -1 } A Q \) in-place, and so, using the basis-changing simulation described in \cref{sec:InPlaceMatrixBasisChange}, multiplication by \( A \) can be performed in-place.
\end{proof}

\subsection{Applications}\label{sec:InPlaceMatrixApplications}

We can now extend \cref{thm:mat-vec-prod} in two ways. First,
by decomposing a matrix into $n$ vectors, we clearly can do matrix-matrix
product by $n$ applications of \cref{thm:mat-vec-prod} directly.

\matmatprod*
\begin{proof}
    To compute $B \mapsto AB$, simply apply the algorithm from \cref{thm:mat-vec-prod} to each column of $B$ one-at-a-time. To compute $B \mapsto BA$, we first replace $B$ with $B^T$, use the previous algorithm to replace $B^T$ with $A^TB^T$, and then take the transpose to obtain $(A^T B^T)^T = BA$. (Note that matrix transpose is in $\inplaceFL$ using the swapping algorithm from \cref{lem:swap}.
\end{proof}

A more intriguing extension is to invert a matrix \( A \) in-place.
We show how to do this in \( \inplaceFCL \).
Our algorithm requires the catalytic tape to contain an invertible matrix \( B \) the same size as \( A \). This requires two preprocessing steps which we defer to the appendices.

First, it relies on the notion of a \emph{field-catalytic subroutine}, by which we mean a uniform family of Turing machines where the catalytic tape stores field elements rather than bits.
This is nontrivial because of the possibility that the catalytic tape may start with invalid representations of field elements;
we postpone discussion of this issue (and the definition of \emph{field-catalytic subroutine}) to \cref{sec:invalid}.

Now assuming that we can work over field elements on the catalytic tape, the second step is to process the tape to ensure it has an invertible matrix on it.
We prove the following lemma in \cref{app:matrix-compress}; the proof involves compressing non-invertible matrices and contains techniques that may be of independent interest.
\begin{lemma}[Putting an invertible matrix on the catalytic tape]\label{lem:invertible-mat-total}
    There exist field-catalytic subroutines \( C,D \) which perform as follows for any field \( K \) representable in \( b = O(\log n) \) bits.
    Let \( t = \lceil \frac{ \log n }{ \log |K| } \rceil + 1 \),
    \( C \) transforms a vector \( \tau \in K^{ t n^2 } \) on its catalytic tape in-place into a vector \( C( \tau ) \in K^{ t n^2  } \) whose first \( n^2 \) coordinates are the entries of an \( n \times n \) invertible matrix, and additionally produces an output \( \mathop{\mathsf{key}}( \tau ) \in \zo^{ O( \log n ) } \).
    \( D \), when supplied with the string \( \mathop{\mathsf{key}}( \tau ) \), transforms the catalytic tape from \( C( \tau ) \) back to \( \tau \) in-place.
\end{lemma}

We are now ready to give our in-place inversion algorithm:
\inplaceinvert*
\begin{proof}
    Below, we describe an algorithm that stores field elements rather than bits on the catalytic tape.
    \Cref{sec:invalid} makes the meaning of this precise: in particular, \cref{lem:SimulateFieldTape} shows that the below algorithm can be converted to work with a catalytic tape that stores bits rather than field elements.

    We first apply \cref{lem:invertible-mat-total} to obtain
    an invertible matrix \(B \in K^{n \times n}\) on the catalytic tape.

    By \cref{cor:mat-mat-prod}, the following
    process can be performed in $\FCL$:
    \begin{enumerate}
        \item replace $\langle A,B \rangle$ with $\langle AB,B \rangle$
        \item replace $\langle AB,B \rangle$ with $\langle AB,B(AB)^{-1} \rangle = \langle AB,A^{-1} \rangle$
        \item replace $\langle AB,A^{-1} \rangle$ with $\langle A^{-1}AB,A^{-1} \rangle = \langle B,A^{-1} \rangle$
        \item swap $B$ and $A^{-1}$ using \cref{lem:swap} to obtain $\langle A^{-1},B \rangle$
    \end{enumerate}
    Note that the second step uses the fact that for any given matrix $C$---in our case
    $C = AB$ is written in memory---we can compute any entry of $C^{-1}$ in $\FCL$ by
    \cref{lem:LinAlg}.

    Since the matrix \( B \) was restored to the catalytic tape, \cref{lem:invertible-mat-total} allows us to then restore the catalytic tape to its original state.
\end{proof}

\section{A relativization barrier to \ts{$\CL \subseteq \P$}{CL ⊆ P}}\label{sec:CLinP}

In our final section, we give an oracle $O$ such that $\CL^O =\EXP^O$.
Our main tool will be ideas from hypercube routing, similar to those used in \cref{sec: fp_in_place} (and see the discussion before \cref{lem: valiant}). We construct an oracle that provides useful
in-place transformations for our catalytic tape while simultaneously
not revealing too much information to an exponential-time machine.

\begin{theorem}\label{thm:exptime-oracle}
    There exists an oracle $O$ such that $\CL^O = \EXP^O$.
\end{theorem}
\begin{proof}
    We will define, for every $d \in \N$, a pair of oracles $S^{(d)}$ (the ``successor'' oracle) and $P^{(d)}$ (the ``password'' oracle), and we will let $O$ be the infinite
    union of all of these oracles.
    
    First, our successor oracles are
    \[ S^{(d)}: \zo^{d} \times \zo^{100 \log_2(d)} \rightarrow ([d] \times \zo^{100 \log_2(d)}) \cup \{ \bot \}\]
    One should imagine $S^{(d)}$ as implicitly defining a directed graph $G^{(d)}$ on $\zo^{d} \times \zo^{100 \log_2(d)}$ with outdegree 1. $S^{(d)}(x, x') = \bot$ indicates that $(x, x')$ has a directed edge to itself, while $S^{(d)}(x, x') = (i, y)$ means that $G^{(d)}$ has a directed edge from $(x, x')$ to $(x+e_i, y)$ where $e_i$ is the standard basis vector $(0, \dots, 0, 1, 0, \dots, 0)$.

    Now given the graph $G^{(d)}$, we can define the password oracles
    \[ P^{(d)}: \zo^{d} \times \zo^{100 \log_2(d)} \times \zo^* \times \zo^* \times \zo^* \rightarrow \{ 0, 1, \bot\} \]
    The password oracle takes as input the name of a vertex $v \in \zo^{d} \times \zo^{100 \log_2(d)}$, the description of an oracle Turing machine $M$, an input $x$, and a number of timesteps $t$.
    When $v$ is a ``valid password'', and $t$ is ``sufficiently small'', the password oracle simulates $M$'s computation with an $O$ oracle on input $x$ for $t$ timesteps (our construction will ensure that this definition is not circular) and, if $M$ halted, outputs its answer. If $v$ is not a valid password, or $M$ does not halt after $t$ steps, or if $t$ is too large, $P^{(d)}$ returns $\bot$.

    To understand the role these oracles play, we define the following game
    over the graph $G^{(d)}$:

    \begin{definition}[Cycle-hiding game]\label{def:routing-game}
        Let $Q^n$ denote the hypercube on $n$ vertices and $K_n$ denote the clique on $n$ vertices. The \emphdef{cycle-hiding game} is played by two players on graph
        $G^{(d)} \subseteq Q^d \times K_{d^{100}}$, with the subset
        $S = \zo^d \times \{0^{100 \log d}\}$ designated as \emphdef{start vertices}.
        In each round of a game, the first player, whom we call the \emphdef{cycle hunter},
        selects a previously unselected vertex, and the second player, whom we
        call the \emphdef{cycle hider}, must respond by either 1) naming one of that
        vertex's neighbors as its ``successor'', or 2) declaring that it has no successor.
        After $R$ rounds of this game, the cycle hunter wins if either 
        \begin{enumerate}
            \item the cycle-hider has revealed a cycle in the successor relations, or
            \item there exists no way for the cycle-hider to define successor relations for the remaining vertices while ensuring that every start vertex belongs to a cycle containing no other start vertices. 
        \end{enumerate}
    \end{definition}
    
    We explain the intuition behind \cref{def:routing-game}.
    The start vertices correspond to the starting configurations
    $s_\tau = \langle \tau, 0^{100 \log d} \rangle$ of the catalytic machine.
    The successor oracle $S^{(d)}$ will define the cycle-hider's answers, which
    we restrict to changing at most one bit on the catalytic tape, plus the entire
    work tape, in each step, giving $1 + 100 \log d$ changes in total.
    If the cycle-hider always answers in a way
    that is consistent with each start vertex lying on its own cycle,
    this means that each $s_\tau$ lies on its own cycle defined by $S^{(d)}$,
    and thus applying $S^{(d)}$ enough times in succession,
    we will eventually reach $s_\tau$ and thus reset the catalytic tape.
    Furthermore, while traversing this cycle we want $s_\tau$ to eventually
    reach a password that is accepted by $P^{(d)}$, but which is not discovered
    by any $\EXP$ machine; hence the goal of the cycle-hider
    is to never actually reveal a complete cycle, giving us room to hide
    at least one password per starting tape $\tau$.

    \begin{lemma}\label{lem:routing-game}
        For sufficiently large $d$, if the number of rounds is less than $2^{d/100}$, then the cycle-hider has a winning strategy in the cycle-hiding game on $G^{(d)}$.
    \end{lemma}

    We now go through the formal definition of the oracle $O$
    assuming \cref{lem:routing-game}, which we prove at the end of this section.
    We will imagine an infinite number of instances of the cycle-hiding game of \cref{lem:routing-game} being played simultaneously, one for each $d \in \N$. The cycle-hunter's strategy will be determined by an enumeration of all oracle Turing machines, and the cycle-hider's responses, as given by the winning strategy guaranteed in \cref{lem:routing-game}, will let us iteratively construct $S^{(d)}$ and $P^{(d)}$.

    We keep an infinite number of ongoing sets $B^{(d)}$ (for bad queries), one for each $d \in \N$. We consider a dovetailing enumeration of all oracle Turing machines: first, simulate the lexicographically first oracle Turing machine for one step, then the lexicographically first two oracle Turing machines for two steps each, and so on. (Here, we consider the input as part of the machine's description.)
    
    Every time a machine makes an oracle query, it makes a query $q$ to either $S^{(d)}$ or $P^{(d)}$ for some $d \in \N$. Say first that the Turing machine queries $S^{(d)}$. If game $d$ is over, we continue the simulation of that Turing machine assuming its oracle call returned $S^{(d)}(q)$ (we will see soon that after game $d$ is over, $S^{(d)}$ is completely fixed). If game $d$ is ongoing and the Turing machine makes an oracle query to $S^{(d)}$ we consider this as a move in the game corresponding to the $d$, and the cycle-hider responds accordingly, either answering using an already computed value of $S^{(d)}$ or fixing $S^{(d)}$'s response on that query.
    
    Now consider the case when the Turing machine makes a query $q$ to $P^{(d)}$. If game $d$ is not over, we add $q$ to $B$ and continue the simulation of the Turing machine assuming the oracle query returned $\bot$. If game $d$ is over, we assume the oracle call returns $P^{(d)}(v, M, t, x)$.
    
    Once $2^{d/100}$ distinct queries have been made in a particular game $d$, we consider the game as over. Since the cycle-hider followed a winning strategy, there must now exist some way to complete the definition of $S^{(d)}$ and ensure that every start vertex $v \in \zo^d \times 0^{100 \log_2(d)}$ is on a unique cycle, while also ensuring that each cycle has a vertex that has not been queried thus far.
    We fix $S^{(d)}$'s description to be such a completion. Note that for every $d \in \N$, $S^{(d)}$ will eventually be fixed since for every input to $S^{(d)}$, there is a Turing machine which eventually queries it, so game $d$ must eventually end.

    Once game $d$ concludes, we also fix the oracle $P^{(d)}$. We first set $P(q) = \bot$ for all $q \in B$. This is what allows our definition of $O$ to be non-cyclic. We now define $P(q)$ for $q \notin B$.
    $P^{(d)}$'s ``passwords'' will consist of all vertices in the game $d$ that have not yet been queried by any machine in the enumeration by the time game $d$ ends.
    (Observe that, since they haven't been queried yet, we are free to fix $P$'s oracle responses arbitrarily on these vertices without impacting any of the previous gameplay.)
    If $v$ a valid password, and $M$ on input $x$ has already been simulated for at least $t$ timesteps in the dovetailing enumeration by the time game $d$ ends, then $P^{(d)}(v, M, t, x)$ will output the state of machine $M$ at time $t$ (i.e. \texttt{Accepting}, \texttt{Rejecting}, or \texttt{Not Yet Halted}). Otherwise, $P^{(d)}(v, M, t, x)$ will output $\bot$. 
 
    This completes the description of the oracle $O$; it now remains to show that $\CL^O = \EXP^O$. Note that since the proof that $\CL \subseteq \ZPP \subseteq \EXP$ is relativizing,
    we have $\CL^O \subseteq \EXP^O$ (under the standard oracle definition \cref{def: standard_oracle_def}). It therefore suffices to show that, for any $c$, and any $\TIME[2^{n^c}]^O$ machine $M$, there exists a $\CL^O$ algorithm deciding the same language $L(M)$.
    Since a $\TIME[2^{n^c}]^O$ machine is also a $\TIME[2^{n^{c'}}]^O$ machine for any $c' > c$, we can assume without loss of generality that $c > |M|$, the description length of the machine (not including the length of its input).

    Let $d = 300 n^c$. Our catalytic algorithm will use $d + 100 \log (d)$ catalytic space and $O(\log d)$ work space on inputs of length $n$. The algorithm works as follows. In the following, $c$ denotes the catalytic space and $w$ denotes the first $100 \log(d)$ bits of workspace.
    \begin{algorithm}
        \caption{$\CL^O$ algorithm for $L(M)$}\label{alg: relativized_cl_algo}
        $a \gets 0$\;
        \Do{$w \neq 0^{100 \log(d)}$}
        {
            $(i, y') \leftarrow S^{(d)}(c, w)$\;
            $c \leftarrow c+1_i$\;
            $w \leftarrow y'$\;
            \If{$P^{(d)}(c, w, M, 2^{n^c}, x) \neq \bot$}
            {
                $a \leftarrow P^{(d)}(c, w, M, 2^{n^c}, x)$\;
            }
        }
        \Return $a$\;
    \end{algorithm}

    Observe that our \cref{alg: relativized_cl_algo} is implicitly taking a walk on the implicitly defined graph $G^{(d)}$ and our algorithm terminates when the walk arrives at any start vertex $s \in \zo^d \times 0^{100 \log d}$.
    
    We first show that \cref{alg: relativized_cl_algo} is in fact a catalytic algorithm. \cref{alg: relativized_cl_algo} clearly only requires $O(\log d) = O(\log n)$ auxiliary work space since $i$ and $y$ are $O(\log d)$ size values. Furthermore, all operations can be done using only $O(\log d)$ extra work space. We now show that \cref{alg: relativized_cl_algo} terminates and returns the catalytic tape to its original position when it terminates. By construction of $S^{(d)}$, and therefore $G^{(d)}$, any walk starting from $(c, 0^{100 \log(d)})$ will eventually reach $(c, 0^{100 \log(d)})$ again. Therefore, our algorithm eventually terminates. Similarly, the fact that it returns the catalytic tape to its original position follows from the fact that no walk starting from $(c, 0^{100 \log(d)})$ ever reach a start vertex $s \in \zo^d \times 0^{100 \log d}$ other than $(c, 0^{100 \log(d)})$.

    We now show that \cref{alg: relativized_cl_algo} outputs $L(M)$. Observe that, by our construction of $S^{(d)}$ and therefore $G^{(d)}$, $(c, w)$ is on a cycle containing at least one password vertex. This follows from the fact that all vertices not queried by the cycle hunter in the game $d$ were deemed valid passwords and the cycle hunter did not find a cycle. Therefore, our catalytic algorithm is guaranteed to find a password. In other words, there will exist a point in \cref{alg: relativized_cl_algo} where $P^{(d)}(c, w, M, 2^{n^c}, x)$ such that $c, w$ is a valid password.

    All that remains to be shown is that when $c, w$ is a valid password, $P^{(d)}(c, w, M, 2^{n^c}, x) \neq \bot$. It is sufficient to show that when game $300 n^{c}$ is declared over, $M$ has already been simulated for at least $2^{n^c}$ steps.
    Observe that, when $M$'s $2^{n^c}$th step is simulated, the total number of oracle queries made thus far in the enumeration can be at most $\max\paren{2^{n^c}, 2^{|M| + |x|}}^2$ since the dovetailing enumeration run for at most that many iterations.
    Since $c > |M|$, we have $n^c > |M| + n$, so the number of oracle queries made is at most $\paren{2^{n^{c}}}^2 = 2^{2n^{c}}$. 
    The game doesn't conclude until $2^{300 n^{c}/100} = 2^{3n^c}$ queries have been made, so $t$ is a ``sufficiently small'' time bound for oracle $P^{(300n^c)}$ to answer the query $(v, M, t, x)$, meaning that our $\CL^O$ algorithm can successfully simulate $M$ on $x$.
\end{proof}

Lastly, we prove \cref{lem:routing-game}, giving a strategy for the cycle-hider
on $G^{(d)}$. To do so, we utilize some ideas from the network routing literature.
\begin{definition}
    We define the bit-fixing path from $x \in \zo^{d}$ to $y \in \zo^{d}$ as the path $x_1 x_2 x_3 \dots x_{d} \rightarrow y_1 x_2 x_3 \dots x_{d} \rightarrow y_1 y_2 x_3 \dots x_{d} \rightarrow \dots \rightarrow y_1 y_2 y_3 \dots y_{d}$.
\end{definition}

The following fact, \cref{lem: valiant}, is a classic result of Valiant.  
In Valiant's setup, we have a network which is a $d$ dimensional hypercube where the vertices are people/servers and connections are edges. Each vertex $v$ wishes to transmit a packet from itself to a different vertex $f(v)$ using connections in the hypercube. The goal is to find a routing scheme where packets are delivered from $v$ to $f(v)$ for all $v \in \zo^d$ but at any given point in time, congestion --- the maximum number of packets at any given vertex --- is minimized. Valiant's insight is that routing each packet to a random vertex (via the bit-fixing path) before routing it to its destination (via the bit-fixing path) yields low congestion, only $O(d)$. As this scheme takes $O(d)$ time steps, we are furthermore guaranteed that at most $O(d^2)$ packets pass through any given vertex across all timesteps.

\begin{lemma}[\cite{ValiantBoolDifficult}]
    \label{lem: valiant}
    Let $f: \zo^d \rightarrow \zo^d$ be a permutation and $N(a, b, c)$ be the multi-set of all vertices occurring on the bit-fixing path from $a \in \zo^d$ to $b \in \zo^d$ and the bit-fixing path from $b$ to $c \in \zo^d$. Let $r: \zo^d \rightarrow \zo^d$ be drawn uniformly from all functions from $\zo^d$ to $\zo^d$. Let us now define $S$ to be the following multi-set.
    \[ S = \bigcup_{v \in \zo^d} N(v, r(v), f(v)) \]
    With $1-\negl(d)$ probability over the choice of $r$, no element occurs in $S$ more than $O(d^2)$ times.
\end{lemma}

This randomized routing trick is the key ingredient in defining a winning strategy for the cycle-hider.

\begin{proof}[Proof of \cref{lem:routing-game}]
     It suffices to show a \emph{randomized} strategy for the cycle-hider that wins with positive probability. We will play on a subgraph $G$ of $G^{(d)}$, which we view as the product of a slightly larger hypercube --- say, $Q^{d + 1000}$ --- and a much smaller clique --- say $d^{99}$ vertices. Note that this is a subgraph of $G^{(d)}$;
    to see this, consider the Boolean strings view of vertices in $G^{(d)}$ as
    $\langle \tau, v \rangle \in \{0,1\}^d \times \{0,1\}^{100 \log d}$, and note that
    $G$ corresponds to increasing the length of $\tau$ by 1000 and decreasing
    the length of $v$ by $\log d$. Furthermore, all start vertices in the previous graph are start vertices in the new graph.
    
    We will talk about this as a $(d+1000)$-dimensional hypercube with ``capacity'' $d^{99}$ at each vertex, since we have effectively a hypercube of this dimension with $d^{99}$ copies of each vertex. Formally, for each $x \in \zo^{d+1000}$, the cycle hider finds keeps track of a set $A_{x} \subseteq [d^{99}]$ which is initialized to $\{ 0^{99 \log d} \}$. The capacity of $x \in \zo^{d+1000}$ is $d^{99}-|A_x|$. Whenever we say that we choose $x$ to be the successor of some vertex $y$, what we actually mean is that we choose a vertex from $(x \times \zo^{99 \log d}) \setminus A_x$ to be the successor of $y$. This allows us to essentially only focus on playing the cycle-hiding game only over $\zo^{d+1000}$ as long as we ensure that we never route to any $x \in \zo^{d+1000}$ which has already exceeded its capacity. From this point on, we will only talk about paths and successors over $\zo^{d+1000}$ while being mindful of the capacity of each note.  We do this with the confidence that the reader can infer the translation to choosing paths and successor over $Q_d \times K_{d^{99}}$. We will also occasionally say that two paths collide at a vertex $v$ if those paths both include $v$.

    The first component of the cycle-hider's strategy will be, before the game begins, to select for each start vertex $s$ a random ``head'' from $Q_{d+1000}$ which we refer to as $h(s)$, and declare a chain of successors leading from $s$ to $h(s)$ along the bit-fixing path.
    Similarly, she'll select, for each start vertex $s$, a random ``tail'' $t(s)$ from $Q_{d+1000}$, and declare a chain of successors leading $t(s)$ to $s$ along the bit-fixing path. $h(s)$ and $t(s)$ may change later and at any point during the game, we refer $\{ h(s) : s \in \text{start vertices} \}$ as active heads and $\{ t(s) : s \in \text{start vertices} \}$ as active tails.
    
    The cycle hider now announces all of these choices unprompted to the cycle-hunter. We now analyze, for any vertex $v \in \zo^{d+1000}$, how much of its capacity is used up in this initial step. We will only consider the paths from stating vertices to their heads for now, we call this the forward path for a vertex. \cref{lem: valiant} tells us with all but negligible probability, for all $v$, only $O(d^2)$ of its capacity is used up in this initial setup of forward paths (choosing forward paths can be thought of as choosing a subset of vertices chosen in \cref{lem: valiant}). The exact same logic applies to the backward paths, those from start vertices to tails. The union bound then tells us that with all but negligible probability, the forward and backward paths use up $O(d^2)$ of the capacity for any given vertex.

    Now, whenever the cycle-hunter makes a query involving a vertex $v \in Q_{d+1000}$ (really it's a query from $Q_{d+1000} \times K_{d^{99}}$, but we only use the first $d+1000$ bits), the cycle-hider will do the following:
    \begin{itemize}
        \item If there are multiple active heads or tails within Hamming distance $d/10$ from $v$, forfeit the game.
        \item Otherwise, if there is one active head (resp. tail) within Hamming distance $d/10$ from $v$, choose a new random point from $Q_{d+1000}$, declare that to be the new active head (resp. tail) of the start vertex of the path from with $v$ originates, and reveal a chain of successors along the bit-fixing sequence to (resp. from) that vertex.
        If this path would pass through a vertex with less than $d^{50}$ capacity remaining, forfeit the game.
        
        Finally, if the queried vertex $v$ appeared along the newly-revealed chain of successors, the cycle-hider has already revealed a successor for the cycle-hunters query.
              Otherwise, respond to the query by announcing that the queried vertex has no successors.
              In this case, mark $v$ as ``ruined'', and consider it as having no remaining capacity.
    \end{itemize}

    We claim that, if the cycle-hider uses this strategy, she will with high probability never have to forfeit the game.
    Let us consider three ways in which the cycle-hunter may be forced to forfeit, and bound their probabilities separately.

    \textbf{Multiple heads and tails close to a query.} If the cycle-hunter can make a query capturing multiple heads and tails within distance $d/10$, then the cycle-hider will forfeit.
    However, note that the set of heads or tails that ever appear in a run of the game is a uniformly random size-$2^{d+1} + 2^{d/100}$ subset of $Q_{d+1000}$ since there are $2^d$ active heads and $2^d$ active tails during the initialization phase and only one new active head/tail is chosen at each of the $2^{d/100}$ rounds. We view this uniformly random size-$2^{d+1} + 2^{d/100}$ subset as a code over $\zo^{d+1000}$.
    With overwhelming probability, this random code will have distance at least $d/5$, meaning that there does not exist any pair close enough to be within distance $d/10$ of the same query.

    \textbf{A path collides with a ruined vertex.} Since we consider a vertex to have no remaining capacity once it has been queried, if any bit-fixing path would pass through one of these already-queried vertices, the cycle-hider must forfeit.
    However, note that we only need to consider collisions with ruined vertices at distance at least $d/10$ from active heads and tails: by the same argument as the previous point, we know with high probability no head or tail will ever be created within distance $d/10$ of an already-queried point,
    so the only way a ruined vertex can exist within distance $d/10$ of a head or tail is if that vertex was queried after the head or tail was already chosen. 
    But if it's queried after the head or tail is chosen, we immediately select a new head or tail, and don't mark the queried point as ruined until after routing away. All ruined vertices are therefore far all from active heads and tails with high probability.

    In order to bound the probability that a given bit-fixing path passes through any ruined vertex at distance at least $d/10$ of the start of the path, we can simply union bound.
    Since at most one vertex is ruined per round, there are at most $2^{d/100}$ ruined vertices in the graph.
    The only way the bit-fixing path could pass through a given ruined vertex is if all of the places where the start vertex disagrees with that ruined vertex, the end vertex agrees with it.
    We're choosing the end vertex randomly, so the probability of this occurring is at most $2^{-d/10}$.
    Thus, the probability that this bit-fixing path passes through any ruined vertex is at most $2^{d / 100} \cdot 2^{-d/10} \ll 2^{-d/100}$, so with high probability no such collision occurs at any round.

    \textbf{Too many paths collide at a non-ruined vertex.} The cycle-hider must also forfeit if any vertex $v \in Q_{d+1000}$ becomes involved in $d^{99} - d^{50}$ different bit-fixing paths.
    But even if she sent a random bit-fixing path from \emph{every} vertex of $Q_{d + 1000}$ to another random vertex, \cref{lem: valiant} that this will occur with extremely small probability.

    So, with high probability the cycle-hider will never have to forfeit.
    This process never reveals a cycle or merges two paths together, because there is always sufficient remaining capacity at each vertex in $Q_{d+1000}$.
    What remains is to guarantee that once this process is over, there is a way to connect each path's head to its tail without creating any collisions.

    To guarantee this, we take each head $h$, pick a random vertex $r$, declare a series of successors along the bit-fixing path from $h$ to $r$, and then declare a series of successors along the bit-fixing path from $r$ to $t$, the tail corresponding to $h$. \cref{lem: valiant} tells us that with high probability there will be fewer than $O(d^2)$ collisions at any vertex.
    Some of these routes may pass through ruined vertices --- but since any given one of these random routes has a smaller than $2^{-d/100}$ chance of passing through a ruined vertex, if we simply ignore the routes that pass through ruined vertices, we have still successfully paired off at least half of the head-tail pairs we needed to.
    So, if we repeat \cref{lem: valiant} $d$ times (each time among the unpaired head-tail pairs), we will eventually find non-ruined paths for every head-tail pair, while introducing at most $d \cdot O(d^2) \ll O(d^{50})$ collisions anywhere.
\end{proof}

\section*{Acknowledgements}
We thank Gil Cohen, Dean Doron, and Ryan Williams for helpful discussions. We thank Ryan Williams in particular for suggesting \cref{cor:rsa}.
\printbibliography

\appendix

\section{Handling invalid field elements on the catalytic tape}\label{sec:invalid}

All of our algorithms describe Turing machines with the tape alphabet \( \zo \), but sometimes we want to do computations over some finite field \( K \).
For the most part, this is easily resolved by requiring the field elements to have representations in \( \zo^b \) for some \( b \) (\cref{def:RepresentableField}) --- then we simply partition a section of tape into \( b \)-bit blocks, and think of each block as one field elements.

Algorithms that read and write field elements on the catalytic tape require more care.
For example, if \( K = \GF(3) \) and its elements are represented as \( 00, 01, 10 \in \zo^2 \), then it's not clear what to do if the catalytic tape starts with the bits \( 111111 \) and the algorithm attempts to, for example, add \( 1 \) to the first field element on the catalytic tape.

This is similar to a problem originally faced by Buhrman, Cleve, Kouck\'y, Loff, and Speelman~\cite{buhrman2014computing} when transforming their register program for \( \TC^1 \) circuits into a catalytic algorithm: the catalytic tape will not always start with valid representations of ring elements.
Our approach is similar to theirs: we ignore \( b \)-bit blocks of catalytic tape that contain invalid field elements.
If the tape starts with too many invalid field elements, we first apply a transformation to each \( b \)-bit block, similar to how Cook and Pyne~\cite{catgraph2025} handle this case.

\begin{definition}\label{def:FieldCatalytic}
A \emph{field-catalytic subroutine}\footnote{
    Despite the word \emph{catalytic}, a field-catalytic subroutine doesn't necessarily restore its catalytic tape. Rather, field-catalytic subroutines are used as algorithms which modify the catalytic tape in some desirable way which can later be reversed.
} is a family of Turing machines \( ( M_K ) \) indexed by representable fields \( K \).
Each machine \( M_K \) has \( K \) as the alphabet for its catalytic tape, and the family must be uniform, in the sense that there is a single universal Turing machine \( \mathcal{M} \) which, given a description of a representable field \( K \) (including Turing machines that compute its operations \( \mathsf{ADD}, \mathsf{MULTIPLY}, \mathsf{VALID} \)), outputs a description of the corresponding machine \( M_K \).
\end{definition}

\begin{lemma}\label{lem:SimulateFieldTape}
Any field-catalytic subroutine \( A \) can be simulated by a catalytic subroutine \( A' \), where if \( A \) uses \( \Poly(n) \) cells (field elements) on the catalytic tape and \( O( \log n ) \) bits of work space, then \( A' \) uses \( \Poly(n) \) catalytic space and \( O( \log n ) \) work space, as long as the field \( K \) is representable in \( b = O( \log n ) \) bits.
If \( A \) always restores its catalytic tape to its starting state (a vector in \( K^m \)), then \( A' \) also restores its catalytic tape to its starting state (a string in \( \zo^{ m' } \)).
\end{lemma}

To prove this, we begin with a simple lemma that helps the algorithm \( A' \) prepare the catalytic tape with sufficiently many valid field elements.

\begin{lemma}\label{lem:ValidXor}
For any \( b \in \N \), nonempty \( S \subseteq \zo^b \), and sequence \( a_1, \dotsc, a_{ \ell } \in \zo^b \), there exists some \( x \in \zo^b \) such that \( a_i \oplus x \in S \) for at least \( \left\lceil \frac{ |S| }{ 2^b } \ell \right\rceil \) distinct indices \( i \).
(Here, \( a_i \oplus x \) denotes the \( b \)-bit string whose \( j \)-th bit is the sum of the \( j \)-th bits of \( a_i \) and \( x \) modulo \( 2 \).)
\end{lemma}

\begin{proof}
Let \( T = \{ (i, x) \in [ \ell ] \times \zo^b \mid a_i \oplus x \in S \} \).
For \( x \in \zo^b \), let \( T_x = \{ i \in [ \ell ] \mid (i, x) \in T \} \).
It suffices to show \( |T_x| \ge \lceil \frac{ |S| }{ 2^b } \rceil \) for some \( x \in \zo^b \).

For any fixed \( i \), the set \( \{ a_i \oplus x \mid x \in \zo^b \} \) exactly equals \( \zo^b \), and so it contains every element of \( S \).
It follows that \( |T| = \ell |S| \), and so since \( |T| = \sum_{ x \in 2^b } |T_x| \), there must be some \( x \in 2^b \) such that \( |T_x| \ge \frac{ \ell |S| }{ 2^b } \).
\end{proof}

\begin{proof}[Proof of \cref{lem:SimulateFieldTape}]
Let \( n \) be the length of the input, and let \( K \) be the field, representable in space \( b \).

Let \( m \) be the number of catalytic tape cells algorithm \( A \) uses.
Then \( A' \) will use \( \lceil \frac{ |K| }{ 2^b } \rceil m b \) bits of catalytic space, and operates as follows.

\( A' \) treats its catalytic tape as consisting of \( \lceil \frac{ 2^b m }{ |K| } \rceil \) blocks with \( b \) bits each.
By \cref{lem:ValidXor}, for any starting configuration of the tape, there must exist some \( x \in \zo^b \) such that after replacing each block's initial value \( \tau_i \) with \( \tau_i \oplus x \), at least \( m \) of the tap blocks start with valid encodings of field elements.
\( A' \) can find such an \( x \) by trying each possibility.
It then records the \( x \) it used on its work tape, and uses it to reverse the transformation at the end of the algorithm.

Then, \( A' \) simulates \( A \) by skipping over invalid \( b \)-bit blocks of catalytic tape.
That is, if \( A \) wishes to read from or write to the \( i \)-th tape cell, \( A' \) finds the \( i \)-th tape cell accepted by \( \mathsf{VALID} \) and applies the corresponding operation to that one.
Note that during this simulation, a block never changes from valid to invalid or vice-versa, so for any fixed \( i \), the \( i \)-th cell of the catalytic tape of \( A \) always corresponds to the same \( b \)-bit block on the catalytic tape of \( A' \), and if \( A \) restores its tape, then each tape block of \( A' \) is restored to its starting value \( \tau_i \oplus x \), which \( A' \) then transforms back to \( \tau_i \) at the end of the simulation.
\end{proof}

\section{Compressing non-invertible matrices}
\label{app:matrix-compress}

In this appendix we will state a number of lemmas, chiefly revolving
around a compress-or-random argument for non-invertible matrices,
that can be used to prove \cref{lem:invertible-mat-total}.

\subsection{Compressing non-invertible matrices}

Here we show how to compress non-invertible matrices in-place to save space.

\Cref{lem:compress-non-invert-exact} shows how to compress a single non-invertible matrix from \( n^2 b \) bits to \( n^2 b - 1 \) bits, but only if the field is exactly representable in \( b > 1 \) bits (and in particular, the field's size must be a power of two).

\begin{lemma}[Compressing one non-invertible matrix over certain fields]
\label{lem:compress-non-invert-exact}
    Let $K$ be a field exactly representable in \( b =O(\log n)\) bits, where $b>1$.
    There exist $\inplaceFL$ algorithms $\Comp: \{0,1\}^{n^2 b} \rightarrow \{0,1\}^{n^2 b - 1}$
    and $\Decomp: \{0,1\}^{n^2 b-1} \rightarrow \{0,1\}^{n^2 b}$ such that $\Decomp(\Comp(Q)) = Q$ for any non-invertible matrix $Q \in K^{ n \times n }$.
\end{lemma}
\begin{proof}
    Let $c_1 \ldots c_n$ be the columns of $Q$, and note that
    there exists an index $i \in [n]$ such that $c_{n-i+1}$ is a linear combination of 
    $c_1 \ldots c_{n-i}$ over $K$.
    
    Fix $i$ to be the largest such value, so that $c_1 \ldots c_{n-i}$ are linearly
    independent over $K$.
    Extend this collection of \( n-i \) vectors to a basis
    \( c_1 \ldots c_{ n-i }, e_{ j_1 } \ldots e_{ j_i } \)
    by adding \( i \) standard basis vectors.
    To make sure this choice is deterministic, depending only on the matrix \( Q \),
    at each step \( k \in [i] \),  set \( j_k \) to be the smallest index such that \( e_{ j_k } \) is
    linearly independent from the vectors already chosen.
    
    There are two important properties of the sequence $j_1 \ldots j_i$.

    First, it can be computed in \( \FCL \) given $c_1 \ldots c_{n-i}$.
    To see this, note that for each \( k \), \( j_k \) is the smallest value
    such that the matrix with columns \( c_1 \ldots c_{ n-i } \) followed by the
    first \( j_k \) standard basis vectors has rank at least \( n-i+k \), and
    matrix rank can be computed by \cref{lem:LinAlg}.

    Second, given vector $c_{ n-i+1 } \in K^n$ but with the coordinates at indices $j_1 \ldots j_k$ erased,
    it is possible to recover in-place the vector \( c_i \) in \( \FCL \).
    That is, if a vector \( c' \in K^n \) is written on the catalytic tape, and \( c' \)
    agrees with \( c_{ n-i+1 } \) on coordinates other than \( j_1 \ldots j_k \), it is
    possible to replace \( c' \) with \( c_{ n-i+1 } \) in-place.
    Here is how to do this.
    Using \cref{lem:LinAlg}, we can in \( \FCL \) invert the matrix whose columns are
    \( c_1 \ldots c_{ n-i }, e_{ j_1 } \ldots e_{ j_i } \), and thus compute coefficients
    coefficients \( ( a_k ) \) and \( ( b_{ \ell } ) \) so that
    \( c' = a_1 c_1 + \dotsb + a_{ n-i } c_{ n-i } + b_1 e_{ j_1 } + \dotsb + b_i e_{ j_i } \).
    If all of \( b_1 \ldots b_i \) are zero, then \( c' = c_{ n-i+1 } \), so we're done.
    Otherwise, record on the work tape the index \( k \) of some nonzero \( b_k \), and the value
    \( b_k \) itself, and then subtract \( b_k e_{ j_k } \) from \( c' \) in-place.
    Then start over. Each time this process is repeated, one more coefficient \( b_k \) will
    be set to zero, so the algorithm will finish after at most \( i \) repetitions.

    Our compression algorithm $\Comp$ now works as follows. Given $Q$, we find this first index
    $i$, which we record in binary on the work tape for the moment. Now we erase
    the entries at indices $j_1 \ldots j_i$ of $c_j$, freeing up $b$ bits
    each, for a total of $bi$ bits freed. Each time we erase an entry, we shift the
    rest of the space used to store the matrix \( Q \) forward, so that in the end,
    the first \( bi \) bits (of the space originally used to store \( Q \)) are unused.

    We are not quite done, since although we've freed \( bi \) bits, we cannot recover the
    matrix \( Q \) without knowing the value \( i \), which is currently stored on the
    work tape. We complete the compression as follows. Replace the first \( i+1 \) bits
    with the string \( 0 1^{i-1} 0 \). This serves two purposes. First, the first bit
    is always zero, and thus may be considered truly freed: we have compressed from
    \( n^2 b \) to \( n^2 b - 1 \) bits. Second, the string \( 1^{ i-1 } 0 \) is a prefix
    encoding of the number \( i \). The decompression algorithm can recover \( i \) by
    counting the number of \( 1 \)s that appear before the first \( 0 \). The reason
    we encode \( i \) in unary instead of binary is to guarantee that the encoding
    always takes fewer than \( bi \) bits. (Notice this also requires \( b>1 \), which
    is why this lemma does not work for the field \( K = \GF(2) \).)

    For $\Decomp$, we first infer \( i \) as the number of 1s seen
    before the first 0 (plus one). We can now determine the indices $j_1 \ldots j_i$ using $c_1 \ldots c_{n-i}$,
    and for each index $j_k$ in reverse order, we recompute the $j_k$th entry of $c_{n-i+1}$
    and insert it into the correct location, shifting the rest of $Q$ as needed.
\end{proof}

For other fields, we use a slightly different approach.
We treat the tape as storing field elements rather than bits (as discussed in \cref{sec:invalid}), and we compress a sequence of matrices $Q_1 \ldots Q_s$ rather than a single matrix $Q$.
\Cref{lem:compress-non-invert-other} compresses a sequence of \( m \) non-invertible matrices from \( mn^2 \) field elements down to \( mn^2-1 \) field elements, for sufficiently large \( m \).
The input and output lie on a catalytic tape whose cells store field elements rather than bits, a notion made precise in \cref{def:FieldCatalytic} (field-catalytic subroutine).

\begin{lemma}[Compressing a sequence of non-invertible matrices over any field]
\label{lem:compress-non-invert-other}
    Let \( K \) be a field representable in space \( b = O( \log n ) \).
    Let \( s \geq \frac{ \log n }{ \log |K| } + 1 \).
    There exist field-catalytic subroutines $\Comp, \Decomp$ which operate as follows on a catalytic tape made out of \( s n^2 \) elements of \( K \).
    If the catalytic tape starts with \( s \) non-invertible matrices, then \( \Comp \) changes it in-place in such a way that the last cell of the catalytic tape is \( 0 \).
    Running \( \Decomp \) after \( \Comp \) restores the catalytic tape.
    (If the tape didn't start with non-invertible matrices, there is no such guarantee.)
\end{lemma}
\begin{proof}
    Similar to the proof of \cref{lem:compress-non-invert-exact}, for each \( t \), let \( i_t \) be the largest value such that the first \( n-i_t \) columns of \( Q_t \) are linearly independent.
    We will use a similar approach to \cref{lem:compress-non-invert-exact}, but we are no longer guaranteed to be able to store an index \( i_t \) in less space than it takes to store \( i_t \) field elements, so we need to change something.

    By averaging, there exists an index $i^*$ such that \( i_t = i^* \) for at least \( s/n \ge \log_{ |K| } n + 1 \) of the matrices \( Q_t \).
    We record the index \( i^* \) in our free space, and then encode each $Q_t$ as follows.
    \begin{itemize}
        \item In all cases, the encoding of begins with the first \( n-i^* \) columns of \( Q_t \) copied exactly.
        \item If \( i_t > i^* \), then the encoding continues with the remaining columns also copied exactly, so \( Q_t \) is simply encoded as itself.
        \item If \( i_t < i^* \), then the encoding continues with \( i_t-1 \) ones followed by a zero: a unary encoding of \( i_t \). Then, we copy each remaining column of \( Q_t \), except that in the \( i_t \)-th column, we omit the entries at indices \( j_1, \dotsc, j_{ i_t } \), where those indices are determined in the same way as in the proof of \cref{lem:compress-non-invert-exact}.
        \item The case \( i_t = i^* \) is handled similarly, except we omit the final zero in the unary encoding of \( i_t \).
    \end{itemize}
    Since we save one field element per $Q_t$ which compresses the $i^*$th column, this gives
    us at least \( s / n \ge \lceil \log_{ |K| } n \rceil + 1 \) free field elements.
    We use \( \lceil \log_{ |K| } n \rceil \) of them to encode \( i^* \), and save one field element overall.

    Our decompression $\Decomp$ is as follows. After recording $i^*$, we consider each $Q_t$
    in turn. First, if the first $n-i^*$ columns contain a linear dependence, we
    leave $Q_t$ untouched, as $i > i^*$ for this matrix. Otherwise, we discern $i$
    from field elements following the first \( n-i^* \) columns, knowing that $i \leq i^*$ and thus a
    sequence of $i^*$ ones indicates that $i = i^*$. Now given $i$, we decompress as in
    \cref{lem:compress-non-invert-exact}.
\end{proof}

\subsection{An alternative approach}

We mention a second approach for putting an invertible matrix on the catalytic tape,
which works for just a single matrix (unlike \cref{lem:compress-non-invert-other}),
but is not a pure compression
argument, as it instead may only ``fix'' the given non-invertible matrix.
\begin{lemma}[Compress-or-fix for non-invertible matrices]
\label{lem:compress-non-invert-other-2}
    Let \( K \) be a field representable in space \( b = O( \log n ) \).
    There exists field-catalytic subroutines \( \Comp, \Decomp, \Fix \) which operate as follows
    on a catalytic tape made out of \( n^2 \) elements of \( K \).
    If the catalytic tape starts with a non-invertible matrix \( Q \in K^{ n^2 } \), then 
    running \( \Comp \) will either (a) compress the tape, meaning
    leave the tape in a state where its last tape cell is \( 0 \) and
    running \( \Decomp \) will restore it, or (b) signal failure and
    leave \( Q \) unchanged. In case (b), the algorithm \( \Fix \) will
    change \( Q \) into an invertible matrix \( B \in K^{ n^2 } \) such
    that \( B \) and \( Q \) differ in only a constant number of locations;
    \( \Fix \) will also output those locations and their original values.
\end{lemma}
\begin{proof}
    Let $k \geq 2$ be any constant. We again repeat the algorithm $\Comp$ from
    \cref{lem:compress-non-invert-exact}, but with the compression split into two cases
    based on the index $i$ of the compressible row.
    If $i > k$, we record $i-k$ in unary as \( i-k-1 \) ones followed by a zero,
    and as before, free \( i \) catalytic tape cells, netting at least \( k - 1 \ge 1 \)
    tape cells saved in the compressed representation.
    $\Decomp$ will be the same as before with the new
    interpretation of $i$.
    
    If $i \leq k$, \( \Comp \) signals failure, and now we describe how \( \Fix \) behaves.
    In this case, the first \( n-i \) columns are linearly independent, so we only
    need to fix the last \( i \) columns.
    Each of those columns can be made linearly independent from the previous columns by
    changing at most one entry (since the \( n \) standard basis vectors can't all be
    linear combinations of the previous columns), so \( Q \) can be transformed into
    an invertible matrix by changing at most \( i \le k \) entries in total.
\end{proof}

\subsection{Proof of \ts{\cref{lem:invertible-mat-total}}{constructing invertible matrices}}

Finally, it is straightforward to iterate the above lemmas to prove
\cref{lem:invertible-mat-total}.

\begin{proof}[Proof of \cref{lem:invertible-mat-total}]
    We will complete this proof using \cref{lem:compress-non-invert-other}, but it can also be proved in much the same way using \cref{lem:compress-non-invert-other-2}, or \cref{lem:compress-non-invert-exact} if \( K \) is exactly representable.

    Our catalytic tape will have length \( t n^2 \) (where \( t = \lceil \frac{ \log n }{ \log |K| } \rceil + 1 \) as in the lemma statement).
    We interpret the last \( (t-1) n^2 \) cells as a sequence of  matrices \( Q_1, \dotsc, Q_{ t-1 } \in K^{ n^2 } \), and our goal will be to put an invertible matrix in the first \( n^2 \) cells.

    Initialize variable \( i \) to zero; \( i \) counts how many of the first \( n^2 \) cells we've freed so far.
    Then we repeatedly do the following.
    Check whether any of the matrices \( Q_1, \dotsc, Q_s \) is invertible.
    If so, swap it with the first \( n^2 \) tape cells and stop.
    Otherwise, apply the algorithm \( \Comp \) from \cref{lem:compress-non-invert-other} to the last \( (t-1) n^2 \) tape cells.
    The last tape cell will then be \( 0 \). Swap it with the \( i \)-th tape cell, and then increase \( i \) by one.

    Repeating the above steps, eventually either an invertible matrix will be found, or \( i \) will reach \( n^2 \).
    In the latter case, the first \( i^2 \) tape cells are all zero, and we may simply write the identity matrix (or any other invertible matrix) to those cells.

    We have successfully found an invertible matrix; all that remains is to describe the algorithm \( D \) which restores the tape, and the string \( \mathop{\mathsf{key}}( \tau ) \) which \( C \) passes on to \( D \) as a helper.
    That latter string describes the state \( C \) stopped in: the value of the variable \( i \), and, if one of the matrices \( Q_j \) was swapped with the first \( n^2 \) tape cells, the index \( j \).
    This is sufficient information for \( D \) to undo everything that \( C \) did, with the help of the algorithm \( \Decomp \) of \cref{lem:compress-non-invert-other}.
\end{proof}

\end{document}